\documentclass[review]{elsarticle}

\usepackage{lineno,hyperref}
\modulolinenumbers[5]

\journal{Theoretical Computer Science}

\usepackage{amsmath,amsthm,amssymb,amsfonts,thmtools,thm-restate}

\usepackage[usenames,dvipsnames]{color}
\usepackage[normalem]{ulem}
\usepackage{amsmath}
\usepackage{amssymb}
\usepackage{stmaryrd}
\usepackage{subcaption}
 \usepackage[all]{xy}
\usepackage{bussproofs}
\usepackage{mathtools} 
\usepackage{wasysym}
\usepackage{lscape}
\usepackage{graphicx}
\usepackage{cancel}

\newcommand{\CNA}{\textsf{CNA}}

\newcommand{\link}[2]{\ensuremath{{}^{#1} \backslash_{#2}}}
\newcommand{\chainedlink}[2]{\ensuremath{\backslash_{#1}^{#2}}}
\newcommand{\startchain}[1]{\ensuremath{{}^{#1}}}
\newcommand{\chainend}[1]{\ensuremath{\backslash_{#1}}}

\newcommand{\silent}{\tau}
\newcommand{\noact}{\scriptstyle {\square}}
\newcommand{\noacts}{\scriptstyle \square}
\newcommand{\merges}[2]{\ensuremath{#1 \bullet   #2}}
\newcommand{\restrict}[1]{\ensuremath{(\nu\, #1)}}

\newcommand{\nil}{\mathbf{0}}

\newcommand{\subst}[2]{#1/#2}

\newcommand{\reduce}[1]{\ensuremath{\mathsf{e}(#1)}}
\newcommand{\stretcheq}{\scriptsize \mathrel \vartriangleright \joinrel \mathrel \vartriangleleft}
\newcommand{\blackstretcheq}{\mathrel{\scriptsize {\mathrel \blacktriangleright \joinrel \mathrel \blacktriangleleft}}}
\newcommand{\simnet}{\ensuremath{\stackrel{\stretcheq}{\sim}}}
\newcommand{\simbnet}{\ensuremath{\stackrel{\blackstretcheq}{\sim}}}

\newcommand{\prodsep}{\;\mid\;}
\newcommand{\defeq}{\triangleq}
\newcommand{\setof}[1]{\{\, #1 \,\}}

\newtheorem{theorem}{Theorem}
\newtheorem{lemma}[theorem]{Lemma}
\newtheorem{proposition}[theorem]{Proposition}
\newtheorem{corollary}[theorem]{Corollary}
\newdefinition{definition}[theorem]{Definition}
\newdefinition{remark}[theorem]{Remark}
\newdefinition{example}[theorem]{Example}

\begin{document}

\begin{frontmatter}

\title{A Formal Approach to Open Multiparty Interactions
\tnoteref{mytitlenote}}
\tnotetext[mytitlenote]{Research partially
    supported by  
    the Italian MIUR Project CINA (PRIN 2010LHT4KM), and
    by Universit\`a di Pisa PRA\_2016\_64 Project \emph{Through the fog}
    and PRA\_2018\_66  \emph{DECLWARE: Metodologie dichiarative per la progettazione e il deployment di applicazioni}.}

\author[unipi]{Chiara Bodei}
\ead{chiara@di.unipi.it}

\author[uniss]{Linda Brodo}
\ead{brodo@uniss.it}

\author[unipi]{Roberto Bruni\corref{cor}}
\ead{bruni@di.unipi.it}

\address[unipi]{Dipartimento di Informatica, Universit\`{a} di Pisa, Italy}
\address[uniss]{Dipartimento Pol.Com.Ing.,   Universit\`a di Sassari, Italy}

\cortext[cor]{Corresponding Author}


\begin{abstract}
We present a process algebra aimed at describing interactions that are
\emph{multiparty}, i.e.~that may involve more than two processes and that are \emph{open}, i.e.\ the number of the processes they involve is not fixed or known a priori.
Here we focus on the theory of a core version of a process calculus, without message passing, called Core Network Algebra (\CNA).
In \CNA\ communication actions are given not in terms of channels but in terms of chains of links that record the source and the target ends of each hop of interactions. 
The operational semantics of our calculus mildly extends the one of CCS.
The abstract semantics is given in the style of bisimulation but requires some ingenuity.
Remarkably, the abstract semantics is a congruence for all operators of \CNA\ and also with respect to substitutions, which is not the case for strong bisimilarity in CCS.
As a motivating and running example, we illustrate the model of a simple software defined network infrastructure.
\end{abstract}


\end{frontmatter}



\section{Introduction}

An \emph{interaction} 
is a way in which communicating processes can influence one another.
Interactions in the time of the World Wide Web and of the Internet of Things (IoT) are something more than input and output between two entities. 
Actually, the word itself can be misleading, by suggesting a reciprocal or mutual kind of actions.
Instead, interactions more and more often involve many parties, and actions are difficult to classify under output and input primitives.
This is a common situation when, e.g.\ a client interacts with a website that in turn invokes some services
from other websites. At a certain level of abstraction it is important to know which are the involved
 services, while it is not important how they are contacted. This practice follows
the ``separation of concern'' modelling style, where the modeller is not interested
in the details of ``how'' (with how many synchronisations, for example) the interaction takes place
as long as a specific phase of the overall procedure is concluded with success.
Intuitively, we can imagine an interaction as the
composition of a jigsaw puzzle: all partners provide different pieces that fit together to complete the picture.

Networks have become part of the critical infrastructure of our daily activities (for business, home, social, health, government, etc.) and a large variety of loosely coupled processes have been offered over global networks, as services.
As a consequence, more sophisticated forms of interactions have become common, for which convenient formal abstractions are under investigation. 
In this regard, one important trend in networking is moving towards architectures where the infrastructure itself can be manipulated by the software, as in the Software Defined Networking (SDN) approach~\cite{DBLP:journals/pieee/KreutzRVRAU15}. Software clients can remotely access and modify the control plane, by using standard open protocols such as OpenFlow.\footnote{See, e.g.\ the Open Networking Foundation website \url{http://www.opennetworking.org}.}
In this case, it is therefore possible to decouple the network control from data-flow and from the network topology and to provide Infrastructure as a Service (IaaS) over data-centers, cloud systems and IoT. 

Another example, coming from a completely different research field, is that of complex biological interactions as the ones emerging in bio-computing and membrane systems, where interactions typically involve several compounds and catalysts.

As a consequence, from a foundational point of view, it is strategic to provide the convenient formal abstractions and models to naturally capture these new communication patterns, by going beyond the ordinary binary form of communication, here called dyadic.
These models should be sufficiently expressive to faithfully describe the complex phenomena,
but they have also to provide a basis for the formal analysis of such systems, by offering sufficient mathematical structure and suitable abstraction mechanisms for
tractability. 

We present here a process algebra, called \CNA, which takes interaction as its basic ingredient.
The described interactions are
\emph{multiparty}, i.e.\ they may involve more than two processes and are \emph{open}, i.e.\ the number of the processes they involve is not fixed or known {a priori}.
This is not to be confused with multiparty interactions represented as a global choreography~\cite{DBLP:conf/popl/HondaYC08,DBLP:journals/csur/HuttelLVCCDMPRT16}, whose realisation is still based on dyadic interactions.
Traditionally in process algebras, communication is based on synchronisation send/receive on specific channels.
In \CNA, instead,
communication actions are given not in terms of channels but in terms of links that record the source and the target ends of each hop of interactions. 
Links can be indeed combined in link chains in order to describe how information can be routed across processes before arriving at destination.
Note that links can be combined if they are to some extent ``complementary", i.e.\ if each process contributes with links that are compatible, if not necessary, with the chain of links provided by the other processes.
According to the
puzzle analogy, different parts of a chain can be composed separately, and,
afterwards, assembled by superposition without overlays.

\begin{figure}[t]
\begin{minipage}[b]{.5\linewidth}
\centering\includegraphics[scale=0.3]{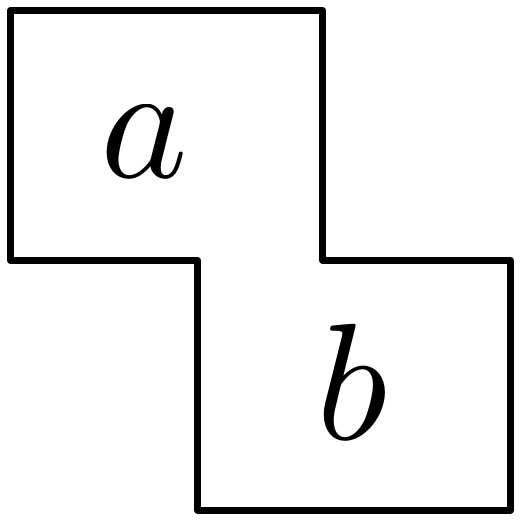}
\subcaption{A link}\label{fig:zlinka}
\end{minipage}%
\begin{minipage}[b]{.5\linewidth}
\centering\includegraphics[scale=0.3]{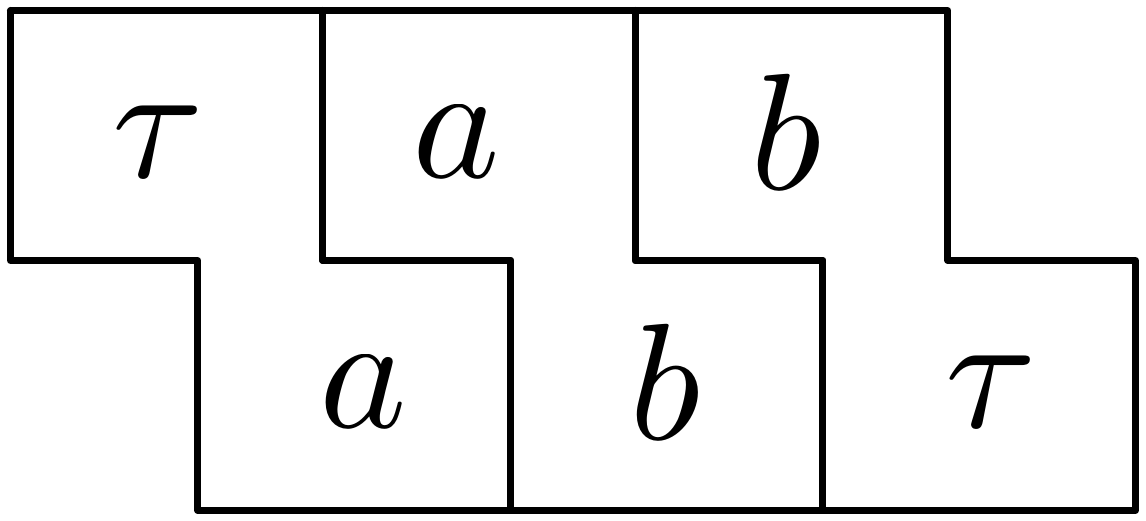}
\subcaption{Matching tetrominos}\label{fig:zlinkb}
\end{minipage}
\caption{Links as tetrominos}\label{fig:zlink}
\end{figure}

To help the intuition, we can see links as Z-shaped tetrominos that can be joined together along a line when the labels of the edges match (see Figure~\ref{fig:zlink}). 

Despite the inherent complexity of representing more sophisticated forms of interaction, we show that the underlying synchronisation algebra and name handling primitives are quite simple, being a straight generalisation of dyadic ones.
This is witnessed by the operational semantic rules of our calculus that, in their simpler version (i.e.\ without message passing), resemble the rules of CCS~\cite{DBLP:books/daglib/0067019}, while in the full one, not considered here (see~\cite{BodeiBB12}), they resemble the ones of $\pi$-calculus~\cite{pi}.
In this sense, \CNA\ processes can be seen as running over a dedicated \emph{middleware} that guarantees a proper handling of links, in the same way as CCS and $\pi$-calculus processes can be seen as running over a middleware that guarantees a proper handling of point-to-point messaging (e.g.\ messages are not lost).

Finally, we address a more technical issue, by providing a convenient abstract semantics, called \emph{network bisimilarity} for \CNA\ processes, which is the analogous of strong bisimilarity for CCS processes. Remarkably, network bisimilarity is a congruence w.r.t.\ all useful composition operators and also w.r.t. \ substitutions, a feature mostly missed in other frameworks.

\paragraph{Synopsis}
In Section~\ref{sec:CCS}, we recall the basics of CCS, although we assume the reader has some familiarity with process algebras. Furthermore, we illustrate a simple scenario of a modular network infrastructure that will serve as a running example to demonstrate that the level of abstraction provided by \CNA\ is much more convenient w.r.t.\ the one provided by processes with dyadic interactions.

In Section~\ref{sec:chains}, we present the theory of link chains, to be used as labels in the operational semantics of \CNA. The theory is quite rich, as it consists of several key operators for building and manipulating link chains. Some nice properties of the introduced operators are also proved, which later will turn out useful to assess semantics properties of \CNA\ processes.

In Section~\ref{sec:network}, we introduce the syntax and operational semantics of \CNA, together with some simple usage examples that should help the reader in understanding the driving principles behind our design choices. 
The key technical contribution in this section is the Accordion Lemma~\ref{lem:stretchlts}.

In Section~\ref{abs-sem}, we close the loop by introducing network bisimilarity, the abstract semantics of \CNA\ processes. 
We prove the main congruence results (see Theorem~\ref{theo:CNAcong} and Proposition~\ref{prop:subs}) and show how network bisimilarity fits well in the running example. Some variations are discussed by the end of Section~\ref{abs-sem}.

Discussion of related work and some concluding remarks are in Section~\ref{sec:conclusion}.

Some auxiliary results and the proofs of technical lemmata can be found in~\ref{app}.

\paragraph{Previous work}
This article is the full version of the extended abstract in~\cite{BodeiBB12}, where also the message passing version of \CNA\ was presented, called \emph{link-calculus}.
Here we focus on the core version of the framework and spell out its theory in full detail. 
It is worth mentioning that we have revised the definition of the equivalence $\stretcheq$ on link chains that is the basis of network bisimilarity and introduced a finer equivalence $\blackstretcheq$ that considerably simplifies the proofs of the main properties.
The motivating and running example is completely original to this contribution.
Finally, we give here all proofs at a good level of detail.

The main contribution in~\cite{BodeiBB12} was to show that the link-calculus can be used to encode Mobile Ambients (MA)~\cite{GC03} in such a way that there is a bijective correspondence between the reduction steps of MA processes and silent transitions of their encodings in the link-calculus. This was a much stronger operational correspondence than any available in the literature, such as the ones in~\cite{Brodo11,B16}.

In~\cite{BodeiBBC14}, following a similar line, 
we have provided an encoding of Brane Calculi~\cite{C05} in the \emph{link-calculus}.
In particular, we have shown that
biologically interactions that
usually involve several compounds 
can be naturally rendered by multiparty interactions.
Furthermore, 
locality can be easily handled, without introducing any specific operator, just
encoding any membrane compartment as a separate process.


\section{Background on CCS and a Running Example}\label{sec:CCS}

\subsection{CCS: the Calculus of Communicating Systems}

The Calculus of Communicating Systems (CCS)~\cite{DBLP:books/daglib/0067019} was introduced by Turing Award winner Robin Milner in the early 1980s. It is based on the notion of processes that communicate on shared channels by executing actions and co-actions over them. Without loss of generality, we can imagine them as input and output actions with synchronous dyadic interaction. 
Let $\mathcal{C} = \{a, b,...\}$ be the set of channels and, by coercion, input actions. 
We denote by $\overline{\mathcal{C}} = \{\overline{a}, \overline{b},...\}$ the set of co-actions (i.e.\ output actions), with $\mathcal{C} \cap \overline{\mathcal{C}} = \emptyset$ and let $\mathcal{O} =  \mathcal{C} \cup \overline{\mathcal{C}}$ denote the set of observable actions
ranged over by $\lambda$. 
We extend the bar-notation to observable actions, by letting $\overline{\overline{\lambda}} = \lambda$.
We also fix a distinguished silent action $\tau \not\in \mathcal{O}$ and let $\mu \in \mathcal{O} \cup\{\tau\}$ denote a generic action.
A \emph{channel relabelling} is a function $\phi:\mathcal{C} \to \mathcal{C}$. It is extended to generic actions by letting $\phi(\tau) = \tau$ and $\phi(\overline{\lambda}) = \overline{\phi(\lambda)}$ for any observable action $\lambda$.  It is called a \emph{renaming} when it is bijective.

A CCS process is then a term generated by the following grammar:
\[
\begin{array}{rcl}
p,q & ::= &
\nil \prodsep
\mu. p  \prodsep
p+q \prodsep
p|q \prodsep
\restrict{a} p \prodsep
p[\phi] \prodsep
A
\end{array}
\]
where $\phi$ is a channel relabelling function and $A$ is any constant drawn from a set $\Delta$ of possibly recursive definitions of the form $A \defeq p$.

Roughly the process $\nil$ is the inactive process that cannot perform any action.
The action prefixed process $\mu.p$ can execute the action $\mu$ and then behaves as $p$.
The operator $+$ introduces nondeterminism: the composed process $p + q$ can behave as $p$ or as $q$, but once it performs an action as $p$ the option $q$ is discarded, and vice versa.
The parallel composition of two processes, written $p | q$,  allows $p$ and $q$ to interleave their actions or to interact by performing complementary actions $a$ and $\overline{a}$: if this is the case, the synchronisation is represented as a silent action $\tau$ and the channel where it takes places is not recorded.
The restricted process $\restrict{a} p$ can perform all actions that $p$ can perform, except for actions $a$ and $\overline{a}$, which are blocked.
The relabelled process $p[\phi]$ can perform all actions that $p$ can perform, but they are relabelled according to $\phi$, i.e. \ if $p$ can do an action $\mu$ then $p[\phi]$ can do $\phi(\mu)$.
Relabelling is very useful for reusing process components in different parts of the system just by changing the set of channels on which they operate.
Finally, the constant $A$ behaves as $p$ if $(A \defeq p) \in \Delta$.

In some cases we shall use constants $A(x_1,...,x_n)$ that are parametric on a set of channel names $x_1,...,x_n$, written more concisely as $A(\widetilde{x})$, and that can be instantiated with actual names, as in $A(a_1,...,a_n)$ or just $A(\widetilde{a})$ for short. Similarly, we write $\restrict{\widetilde{a}}p$ for $\restrict{a_1}\cdots\restrict{a_n}p$.

\begin{figure}[t]
\begin{center} 
\begin{prooftree} 
\AxiomC{$\phantom{\mu.p \xrightarrow{\mu} p}$} 
\UnaryInfC{$\mu.p \xrightarrow{\mu} p$} 
\DisplayProof
\quad
\AxiomC{$p \xrightarrow{\mu} p'$} 
\UnaryInfC{$p+q \xrightarrow{\mu} p'$} 
\DisplayProof
\quad
\AxiomC{$q \xrightarrow{\mu} q'$} 
\UnaryInfC{$p+q \xrightarrow{\mu} q'$} 
\DisplayProof
\quad
\AxiomC{$p \xrightarrow{\mu} p'$} 
\AxiomC{$\mu\not\in\{a,\overline{a}\}$} 
\BinaryInfC{$\restrict{a}p \xrightarrow{\mu} \restrict{a}p'$} 
\end{prooftree} 
\end{center}

\begin{center} 
\begin{prooftree} 
\AxiomC{$p \xrightarrow{\mu} p'$} 
\UnaryInfC{$p[\phi] \xrightarrow{\phi(\mu)} p'[\phi]$} 
\DisplayProof
\quad
\AxiomC{$p \xrightarrow{\mu} p'$} 
\UnaryInfC{$p|q \xrightarrow{\mu} p'|q$} 
\DisplayProof
\quad 
\AxiomC{$q \xrightarrow{\mu} q'$} 
\UnaryInfC{$p|q \xrightarrow{\mu} p|q'$} 
\DisplayProof
\quad 
\AxiomC{$p \xrightarrow{\lambda} p'$} 
\AxiomC{$q \xrightarrow{\overline{\lambda}} q'$} 
\BinaryInfC{$p|q \xrightarrow{\tau} p'|q'$} 
\DisplayProof
\end{center}
\begin{center}
\AxiomC{$p \xrightarrow{\mu} q$}
\AxiomC{$(A \defeq p)\in\Delta$}
\BinaryInfC{$A \xrightarrow{\mu} q$} 
\end{prooftree} 
\end{center}
\caption{SOS semantics of CCS.}
\label{fig:ccssos}
\end{figure}

The operational semantics of CCS is given in the form of a Labelled Transition System (LTS), where the states are CCS processes and the transitions are labelled by actions. We write $p \xrightarrow{\mu} q$ if $p$ can perform the action $\mu$ and behave as $q$ afterwards. The inference rules that generate the LTS are defined in the style of Structural Operational Semantics (SOS) as they are driven by the syntax of processes (see Figure~\ref{fig:ccssos}).

For example, we have transitions such as $a.b.\nil \xrightarrow{a} b.\nil \xrightarrow{b} \nil$, $a.b.\nil + c.\nil \xrightarrow{a} b.\nil$ and $a.b.\nil + c.\nil \xrightarrow{c} \nil$. The interplay between restriction and parallel composition is interesting as it can be used to impose synchronisation on some channel. In fact, while for $A \defeq (a.b.\nil + c.\nil) | (\overline{a}.\nil + d.\nil)$ we have transitions such as 
$A\xrightarrow{a} b.\nil | (\overline{a}.\nil + d.\nil)$, 
$A\xrightarrow{\overline{a}} (a.b.\nil + c.\nil) | \nil$, 
and 
$A\xrightarrow{\tau} b.\nil | \nil$, 
among others, the process $\restrict{a} A$ cannot perform any action labelled by $a$ and $\overline{a}$ but can still perform the synchronisation 
$\restrict{a} A\xrightarrow{\tau} \restrict{a}(b.\nil | \nil)$, because $\tau$ actions cannot be restricted. 

To keep the notation compact, we write $p \xrightarrow{\mu_1} \xrightarrow{\mu_2} \cdots \xrightarrow{\mu_n}q$ when there exist some processes $p_1,...p_{n+1}$ that we do not need to mention such that $p_1 = p$, $p_{n+1} = q$ and $p_i \xrightarrow{\mu_i} p_{i+1}$ for $i\in [1,n]$.

Recursive definitions can be used to account for infinite behaviour.
For example, if $(A \defeq a.A + b.\nil) \in \Delta$, then the process $A$ can do any finite sequence of actions $a$ terminated by an action $b$, as in 
$A \xrightarrow{a} A \xrightarrow{a} \cdots \xrightarrow{a} A \xrightarrow{b} \nil$ but it can also perform an infinite sequence of actions $a$, as in
$A \xrightarrow{a} A \xrightarrow{a} \cdots \xrightarrow{a}  A \xrightarrow{a} \cdots.$

The main notion of equivalence for CCS processes is called \emph{strong bisimilarity} and is denoted by $\sim$.
It is defined as the largest strong bisimulation relation,  i.e.\  as the largest binary relation $\mathbf{R}$ on CCS processes such that whenever $p\mathop{\mathbf{R}} q$ we have that:
\begin{enumerate}
\item for any $\mu,p'$ such that $p \xrightarrow{\mu} p'$ there exists $q'$ such that $q \xrightarrow{\mu} q'$ and $p'\mathop{\mathbf{R}} q'$;
\item for any $\mu,q'$ such that $q \xrightarrow{\mu} q'$ there exists $p'$ such that $p \xrightarrow{\mu} p'$ and $p'\mathop{\mathbf{R}} q'$.
\end{enumerate}

Notably, strong bisimilarity is a congruence w.r.t. all the operators of CCS.
Here we point out that it is not a congruence w.r.t. action substitution.
For example, it is well-known that strong bisimilarity reduces concurrency to non-determinism, as
$a.\nil \ |\ \overline{b}.\nil \sim a.\overline{b}.\nil + \overline{b}.a.\nil$. However, if we apply the (non-injective) substitution $\{\subst{b}{a}\}$ that replaces all the (free) occurrences of $a$ with $b$ to both processes we get $b.\nil \ |\ \overline{b}.\nil \not\sim b.\overline{b}.\nil + \overline{b}.b.\nil$, because the former process can do the silent step $b.\nil \ |\ \overline{b}.\nil \xrightarrow{\tau} \nil\ |\ \nil$, while the latter process $b.\overline{b}.\nil + \overline{b}.b.\nil$ cannot.

Sometimes one wants to abstract away from silent transitions $\tau$. 
Correspondingly \emph{weak bisimilarity} can be considered instead of strong bisimilarity, where in the bisimulation game a single transition can be simulated by exploiting any number of silent transitions.
Unfortunately, weak bisimilarity is not a congruence w.r.t. the choice operator and substitutions.

\subsection{Software Defined Infrastructures}

In this sub-section we sketch four scenarios of increasing complexity together with their possible modelling in CCS.
Once introduced \CNA, in Sections~\ref{sec:network} and~\ref{abs-sem}, we will revisit these examples to show that they can be more conveniently accounted for in \CNA.

The reference case study consists of a 
network infrastructure with $n$ requestor agents $A_1,...,A_n$, $m$ servers $S_1,...,S_m$ and a routing infrastructure $R$ that regulates which requestors are connected to which servers in a way that is out of the control of agents and requestors.
For the sake of simplicity, in the following we let $n=m=2$.

\begin{example}[Blind routing]\label{blind}
Initially, we keep the scenario as simple as possible: the idea is that a requestor can repeatedly request a service if there is a non-busy server connected to it via $R$. Let us suppose that $R$ connects $A_1$ with $S_1$ and $S_2$, while $A_2$ only with $S_2$.
In CCS, the system can be readily modelled by the following recursive processes that run in parallel.
\begin{eqnarray*}
A_i & \defeq & \overline{\mathit{req}_i}.\overline{\mathit{think}}.A_i\quad \mbox{ for $i\in[1,2]$} \\
S_j & \defeq & \mathit{srv}_j.\overline{\mathit{exec}}.S_j + \overline{\mathit{busy}}.\tau.S_j\quad \mbox{ for $j\in[1,2]$}\\
R & \defeq & \mathit{req}_1.(\overline{\mathit{srv}_1}.R + \overline{\mathit{srv}_2}.R) + \mathit{req}_2.\overline{\mathit{srv}_2}.R
\end{eqnarray*}

\noindent
For example we can let the system be defined as

\[
N \defeq 
\restrict{\widetilde{\mathit{req}}}
\restrict{\widetilde{\mathit{srv}}} 
(A_1\ |\ A_2\ |\ R\ |\ S_1\ |\ S_2)
\]
so that synchronisation is enforced for all the interactions between requestors and the infrastructure 
(on channels $\mathit{req}_1$ and $\mathit{req}_2$)
 and between the infrastructure and servers 
(on channels $\mathit{srv}_1$ and $\mathit{srv}_2$).

The routing depends on the state of the system. Suppose, for instance,
that $S_2$ becomes busy and that $A_2$ sends a request to $R$, 
according to the transition sequence 
$N \xrightarrow{\overline{\mathit{busy}}} \xrightarrow{\tau} N'$ with
\[N' = 
\restrict{\widetilde{\mathit{req}}}
\restrict{\widetilde{\mathit{srv}}} 
(A_1\ |\ \overline{think}.A_2\ |\ \overline{\mathit{srv}_2}.R\ |\ S_1\ |\ \tau.S_2) .
\]

This is perfectly admissible, but leaves $A_2$ thinking its request has been served, because the interaction with $R$ has taken place, while the server $S_2$ has not even received it.
\end{example}

\begin{example}[Acknowledged routing]\label{ack}
To remedy the problem raised by the previous model, one can introduce some acknowledgement protocol, to ensure each agent that its request has been assigned to some server. The CCS model can thus be improved by redesigning the processes as follows:
\begin{eqnarray*}
A_i & \defeq & \overline{\mathit{req}_i}.\mathit{ack}_i.\overline{\mathit{think}}.A_i\quad \mbox{ for $i\in[1,2]$} \\
S_j & \defeq & \mathit{srv}_j.\overline{\mathit{exec}}.S_j + \overline{\mathit{busy}}.\tau.S_j\quad \mbox{ for $j\in[1,2]$}\\
R & \defeq & \mathit{req}_1.(\overline{\mathit{srv}_1}.\overline{\mathit{ack}_1}.R + \overline{\mathit{srv}_2}.\overline{\mathit{ack}_1}.R) + \mathit{req}_2.\overline{\mathit{srv}_2}.\overline{\mathit{ack}_2}.R
\end{eqnarray*}
For example we can let the system be defined as
\[
M \defeq \restrict{\widetilde{\mathit{ack}}}  N
\]
where $N$ is defined as before.

At a very abstract level, we can view an infrastructure as an oriented graph with $n$ nodes on the left boundary and $m$ nodes on the right boundary: the assignment of a request from $A_i$ to the server $S_j$ is possible if $S_j$ is available and if there is a connection between the $i$th node on the left boundary and the $j$th node on the right boundary. Graphically, $M$ can be depicted as below:
\[
\xymatrix{
*++[F]{A_1} \ar@{.>}@/^/[r]^{\mathit{req}_1} &
{_1\bullet} \ar[r] \ar[rd] \ar@{.>}@/^/[l]^{\mathit{ack}_1} 
\POS[]+<1.5pc,1pc>*{R}
\POS[]+<1.65pc,-1.8pc>*=<2.9pc,7pc>[F-]{\;}
& 
{\bullet_1} \ar@{.>}@/^/[r]^{\mathit{srv}_1} &
*++[F]{S_1} \\
*++[F]{A_2} \ar@{.>}@/^/[r]^{\mathit{req}_2} &
{_2\bullet} \ar[r] \ar@{.>}@/^/[l]^{\mathit{ack}_2} & 
{\bullet_2} \ar@{.>}@/^/[r]^{\mathit{srv}_2} & 
*++[F]{S_2} 
}
\]

This time, when $S_2$ is busy and $A_2$ interacts with the infrastructure $R$ on channel $\mathit{req}_2$, the requestor $A_2$ blocks until the server $S_2$ becomes available and can accept the request by interacting on channel $\mathit{srv}_2$ with $R$. In fact, when this is the case, $R$ sends the acknowledgment on channel $\mathit{ack}_2$ to $A_2$.
\end{example}

\begin{example}[Composite, acknowledged routing]\label{composite}
Now suppose that the infrastructure $R$ is not monolithic, and that it is instead  obtained by composing some network infrastructures together, which is a necessity for complex systems. 

To make the infrastructure compositional, we must make interaction symmetric on the two, left and right, boundaries,  i.e.\  we must assume that servers also send some acknowledgement. 
Correspondingly, we set
\begin{eqnarray*}
S_j & = & \mathit{srv}_j.\overline{\mathit{sack}_j}.\overline{\mathit{exec}}.S_j + \overline{\mathit{busy}}.\tau.S_j\quad \mbox{ for $j\in[1,2]$}\\
\end{eqnarray*}

Now the system can be depicted as below:

\[
\xymatrix{
*++[F]{A_1} \ar@{.>}@/^/[r]^{\mathit{req}_1} &
{_1\bullet} \ar[r] \ar[rd] \ar@{.>}@/^/[l]^{\mathit{ack}_1} 
\POS[]+<1.5pc,1pc>*{R}
\POS[]+<1.65pc,-1.8pc>*=<2.9pc,7pc>[F-]{\;}
& 
{\bullet_1} \ar@{.>}@/^/[r]^{\mathit{srv}_1} &
*++[F]{S_1} \ar@{.>}@/^/[l]^{\mathit{sack}_1} \\
*++[F]{A_2} \ar@{.>}@/^/[r]^{\mathit{req}_2} &
{_2\bullet} \ar[r] \ar@{.>}@/^/[l]^{\mathit{ack}_2} & 
{\bullet_2} \ar@{.>}@/^/[r]^{\mathit{srv}_2} & 
*++[F]{S_2} \ar@{.>}@/^/[l]^{\mathit{sack}_2}
}
\]

Now consider the case where $R$ is obtained by juxtaposing three other infrastructures $R'$, $R''$ and $R'''$ defined as follows:
\begin{eqnarray*}
R' & \defeq & \mathit{req}_1.(\overline{s_1}.a_1.\overline{\mathit{ack}_1}.R' + \overline{s_2}.a_2.\overline{\mathit{ack}_1}.R') + \mathit{req}_2.\overline{s_2}.a_2.\overline{\mathit{ack}_2}.R' \\
R'' & \defeq & s_1.\overline{s'_1}.a'_1.\overline{a_1}.R'' + s_2.\overline{s'_2}.a'_2.\overline{a_2}.R' \\
R''' & \defeq & s'_2.\overline{\mathit{srv}_2}.\mathit{sack}_2.\overline{a'_2}.R''' \\
R & \defeq & \restrict{\widetilde{a'}}\restrict{\widetilde{s'}}\restrict{\widetilde{a}}\restrict{\widetilde{s}} (R' \ |\ R''\ |\ R''')
\end{eqnarray*}
Note that $R'''$ does not forward any request coming from its first port.
The resulting infrastructure is illustrated in the figure below:
\[
\xymatrix{
*++[F]{A_1} \ar@{.>}@/^/[r]^{\mathit{req}_1} &
{_1\bullet} \ar[r] \ar[rd] \ar@{.>}@/^/[l]^{\mathit{ack}_1} 
\POS[]+<1.5pc,1pc>*{R'}
\POS[]+<1.65pc,-1.8pc>*=<2.9pc,7pc>[F-]{\;}
& 
{\bullet_1} \ar@{.>}@/^/[r]^{s_1} 
&
{_1\bullet} \ar[r] \ar@{.>}@/^/[l]^{a_1} 
\POS[]+<1.5pc,1pc>*{R''}
\POS[]+<1.65pc,-1.8pc>*=<2.9pc,7pc>[F-]{\;}
\POS[]+<1.5pc,2.5pc>*{R}
\POS[]+<1.65pc,-1.8pc>*=<16.3pc,9.6pc>[F--]{\;}
&
{\bullet_1} \ar@{.>}@/^/[r]^{s'_1} &
{_1\bullet} \ar@{.>}@/^/[l]^{a'_1} 
\POS[]+<1.5pc,1pc>*{R'''}
\POS[]+<1.65pc,-1.8pc>*=<2.9pc,7pc>[F-]{\;}
&
{\bullet_1} \ar@{.>}@/^/[r]^{\mathit{srv}_1} &
*++[F]{S_1} \ar@{.>}@/^/[l]^{\mathit{sack}_1} \\
*++[F]{A_2} \ar@{.>}@/^/[r]^{\mathit{req}_2} &
{_2\bullet} \ar[r] \ar@{.>}@/^/[l]^{\mathit{ack}_2} & 
{\bullet_2} \ar@{.>}@/^/[r]^{s_2} &
{_2\bullet} \ar[r] \ar@{.>}@/^/[l]^{a_2} &
{\bullet_2} \ar@{.>}@/^/[r]^{s'_2} & 
{_2\bullet} \ar[r] \ar@{.>}@/^/[l]^{a'_2} &
{\bullet_2} \ar@{.>}@/^/[r]^{\mathit{srv}_2} & 
*++[F]{S_2} \ar@{.>}@/^/[l]^{\mathit{sack}_2}
}
\]

In general, an assignment of a request to a server is possible only if there is a path of connections in the graph associated with the infrastructure. In the example, the requests coming from $A_1$ and $A_2$ can only be assigned to $S_2$, as there is no path towards $S_1$.

Unfortunately, it may happen that the routing of the infrastructure comes to a dead point. In the example, $R'$ can forward the request from $A_1$ to $R''$ on $s_1$, but then $R''$ is blocked because it will not be able to pass the request to $R'''$ on $s'_1$.

To remedy this, either all dead paths must be removed before the infrastructure is deployed or some deadlock-detection and backtracking mechanism should be put in place, which requires some additional efforts.
\end{example}

\begin{example}[Dynamic routing]
Finally, imagine the situation where the infrastructure $R$ is \emph{software defined}, in the sense that connections can be added and removed dynamically. This time static-time dead-path analysis is not possible at all, and the integration of this additional feature with the previous acknowledgement, deadlock-detection and backtracking mechanisms looks overly complicated.
\end{example}



\section{A Theory of Link Chains}
\label{sec:chains}

To address the challenges posed by the scenarios in Section~\ref{sec:CCS}, the idea is to move from dyadic interaction to multiparty one.
Correspondingly, communication actions are given in terms of {\em links} and a single atomic interaction is   
possibly composed by more than one link.
A link is a pair $\link{\alpha}{\beta}$
that records the source and the target sites of a communication, 
meaning that the input available at the source site $\alpha$ can be forwarded to the target one $\beta$.  
Links are suitably combined in link chains to describe how information can be routed across processes before arriving at their destination.
Therefore, links are combined like pieces in a jigsaw puzzle, where each party contributes with its link. As explained in the introduction, we can think about links as Z-shaped tetrominos that are joined together in a line when the labels of the edges match so to form a link chain.
Standard 
I/O communication is made more accurate, by recording the route of information across several sites.
Furthermore, link chains allow seamless realisation of multiparty synchronisations.

To achieve compositionality, we allow processes to provide link chains that are assembled just in part.
Intuitively they correspond to puzzles where some, but not all, the pieces are present.
As an example, Figure~\ref{fig:chain} shows a chain with a missing link from $a$ to $b$.

\begin{figure}[t]
\begin{center}
\includegraphics[scale=0.3]{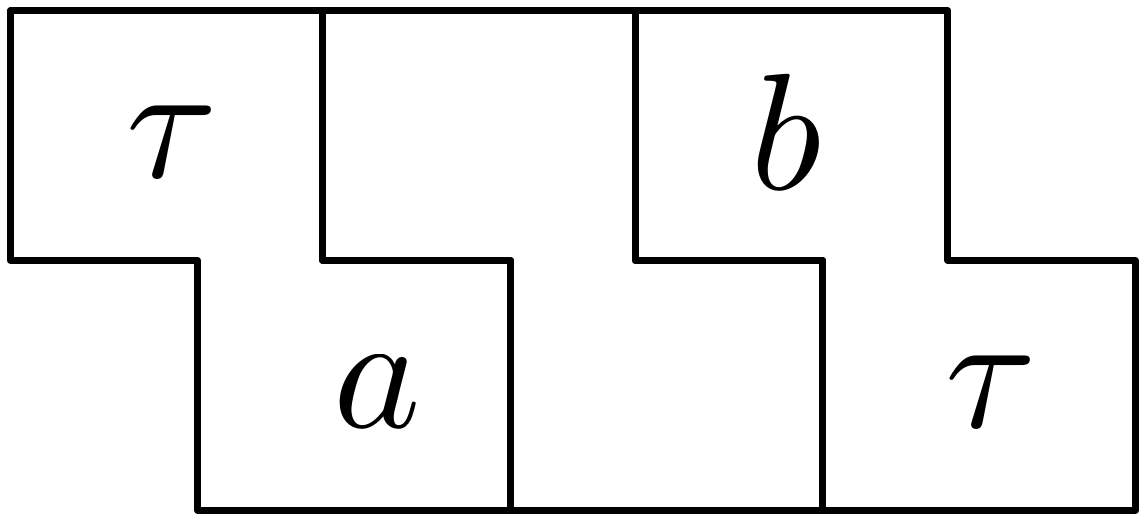}
\end{center}
\caption{A chain with a missing link.}
\label{fig:chain}
\end{figure}

In this section we present the underlying theory of links and link chains, posing the emphasis on the operations for combining them and on some relevant properties they satisfy.


\subsection{Links}
Let $\mathcal{C}$ be the set of channels, ranged over by $a,b,c,...$, and 
let $\mathcal{A} = \mathcal{C} \cup \setof{\silent} \cup \setof{{\noacts}}$ be the set of actions, ranged over by $\alpha,\beta,\gamma,...$,
where the symbol $\silent$ denotes a \emph{silent} action, while the symbol $\noact$ denotes a \emph{virtual} (non-specified) action (i.e. \ a missing piece of the puzzle according to the analogy proposed above).

\begin{definition}[Links: solid, virtual, valid]
A \emph{link} is a pair $\ell=\link{\alpha}{\beta}$;
 it can be read as forwarding the input available on $\alpha$ to $\beta$, and we call $\alpha$ the \emph{source site} of $\ell$ and $\beta$ the \emph{target site} of $\ell$.
A link $\link{\alpha}{\beta}$ is \emph{solid} if $\alpha,\beta \neq \ \noacts$; 
the link $\link{\noact}{\noact}$ is called \emph{virtual}.
A link is \emph{valid} if it is solid or virtual.
We let $\mathcal{L}$ be the set of valid links.
\end{definition}

Examples of non valid links are $\link{\tau}{\noact}$ and $\link{\noact}{a}$, while both $\link{\tau}{a}$ and $\link{a}{b}$ are solid valid links.
From now on, we only consider valid links.

As it will be  shortly explained, the virtual link $\link{\noact}{\noact}$ is a sort of ``missing link'' inside a link chain; it represents a needed link that can be supplied, as a solid link, by another link chain, via a suitable composition operation called \emph{merge} (see below).


\subsection{Link Chains}
Links can be combined in link chains that record the source and the target 
sites of each hop of the interaction.

\begin{definition}[Link Chain]
A \emph{link chain} is a 
finite sequence $s = \ell_{1}...\ell_{n}$ of (valid) links  $\ell_{i} = \link{\alpha_{i}}{\beta_{i}}$ such that:
\begin{enumerate}

\item for any $i\in [1,n-1]$, 
$\left\{\begin{array}{ll}
\beta_{i},\alpha_{i+1}\in \mathcal{C} & \mbox{ implies } \beta_{i} = \alpha_{i+1}\\
\beta_{i}=\silent & \mbox{ iff } \alpha_{i+1}=\silent
\end{array}\right.
$
\item  $\exists i \in [1,n]. \ \ell_i \neq \link{\noact}{\noact}$.

\end{enumerate}
\end{definition}

The first condition says that any two adjacent solid links must match on their adjacent sites;
it also imposes that, in particular, $\silent$ cannot be matched by $\noact$.
The second condition disallows chains made of virtual links only.
A non-empty link chain is \emph{solid} if all its links are so.
For example, $\startchain{\silent}\chainedlink{a}{a}\chainend{b}$ is a solid link chain, 
while $\startchain{\silent}\chainedlink{a}{\noact}\chainedlink{\noact}{b}\chainend{\silent}$ is not solid.

%

In counting links in a chain we may decide to ignore or not virtual links. 

\begin{definition}[Length and size]
The \emph{length} of a chain $s$, written $|s|$, is the number of valid (virtual and solid) links that are in $s$.
The \emph{size} of $s$, written $||s||$, is the number of solid links that are in $s$.
\end{definition}

For example, $|\startchain{\silent}\chainedlink{a}{\noact}\chainedlink{\noact}{b}\chainend{\silent}| = 3$, while
$||\startchain{\silent}\chainedlink{a}{\noact}\chainedlink{\noact}{b}\chainend{\silent}||= 2$.

The following definition introduces an equivalence relation over link chains that equates two valid link chains if they only differ for the presence of virtual links only.

\begin{definition}[Equivalence $\blackstretcheq$]\label{def:black}
We let $\blackstretcheq$ be the least equivalence relation 
over link chains closed under the axioms (whenever both sides are well defined link chains):
\[
\begin{array}{rclcrcl}
s\link{\noact}{\noact} & \blackstretcheq &  s & \qquad &
s_1 \startchain{\noact}\chainedlink{\noact}{\noact}\chainend{\noact}s_2 & \blackstretcheq & s_1 \link{\noact}{\noact} s_2\\
\link{\noact}{\noact}s & \blackstretcheq & s & &
s_1 \startchain{\alpha}\chainedlink{a}{a}\chainend{\beta} s_2
& \blackstretcheq & 
s_1 \startchain{\alpha}\chainedlink{a}{\noact}\chainedlink{\noact}{a}\chainend{\beta}s_2 
\end{array}
\]
\end{definition}

From the above definition, it follows that the chain size is invariant w.r.t.~$\blackstretcheq$,  i.e.\  $s \blackstretcheq s'$ implies that $||s|| = ||s'||$ (although $s$ and $s'$ can have different lengths).
Furthermore, for $\ell$ a solid link and $s$ a link chain, we write $s \blackstretcheq  \ell$ if and only if $\ell$ is the {\em only} solid link that occurs in $s$.

The following basic operations over links and link chains are partial and strict, i.e.\ they may issue $\bot$ (undefined) and the result is $\bot$ if any argument is $\bot$. 
To keep the notation short, we tacitly assume that
the result is $\bot$ if either one of the sub-expressions in the righthand side of any defining equation is undefined, or
if none of the conditions in the righthand side of any defining equation is met.

\paragraph{Merge} 
We remind that the virtual links in a chain can be seen as the part in the chain not yet specified, and possibly provided by another link chain when merged. 

Two link chains can be merged if they are to some extent ``complementary'', in the sense that: (i)~they have the same length; (ii)~each of them provides solid links that are missing in the other, and (iii)~superimposed together they still form a link chain. 

In particular, if there is a position where both link chains carry solid links, then there is a clash and the merge is not possible (undefined).
Also if the merge would result in a non valid sequence, then the merge is not possible.

\begin{definition}[Merge]
For $s = \ell_{1}...\ell_{n}$ and $s' = \ell'_{1}...\ell'_{n}$, with $\ell_{i} = \link{\alpha_{i}}{\beta_{i}}$ and $\ell'_{i} = \link{\alpha'_{i}}{\beta'_{i}}$ for any $i\in [1,n]$, we define their merge $\merges{s}{s'}$ by defining the merge of two actions as follows:
\[
\begin{array}{lll}
\merges{\alpha}{\beta} & \defeq &
\left\{\begin{array}{ll}
\alpha & \mbox{ if } \beta=\ \noacts\\
\beta & \mbox{ if } \alpha=\ \noacts
\end{array}\right.
\end{array}
\]
and then taking its homomorphic extension to links and link chains:\footnote{As anticipated, we remark that in the defining equations for merge it is implicitly understood that:
if $\merges{\ell_{i}}{\ell'_{i}} = \bot$ for some $i$, then $\merges{s}{s'} = \bot$;
if the sequence $(\merges{\ell_{1}}{\ell'_{1}})...(\merges{\ell_{n}}{\ell'_{n}})$ is not a link chain, then $\merges{s}{s'} = \bot$;
if $\merges{\alpha}{\alpha'}=\bot$ or $\merges{\beta}{\beta'}=\bot$, then $\merges{\link{\alpha}{\beta}}{\link{\alpha'}{\beta'}} = \bot$;
if $\alpha,\beta\neq \ \noact$, then $\merges{\alpha}{\beta} = \bot$.}
\[
\begin{array}{rll@{\hspace{3em}} rll}
\merges{s}{s'} & \defeq &
(\merges{\ell_{1}}{\ell'_{1}})\cdots(\merges{\ell_{n}}{\ell'_{n}}) 
& 
\merges{\link{\alpha}{\beta}}{\link{\alpha'}{\beta'}} & \defeq &
\link{(\merges{\alpha}{\alpha'})}{(\merges{\beta}{\beta'})}
\end{array}
\]
\end{definition}


Roughly, the merge is defined element-wise on the actions of a link chain, by ensuring that whenever two actions are merged, (at least) one of them is $\noacts$ and that the result of the merge is still a link chain.
Note that the merge is undefined if the link chains have different lengths.

Intuitively,
we can imagine that $s$ and $s'$ are two parts of the same puzzle separately assembled, where solid links are the pieces of the puzzle and virtual links are the holes in the puzzle and their merge $\merges{s}{s'}$ puts the two matched parts together, without
piece overlaps (see Figure~\ref{tetrom2}).

\begin{figure}[t]
\begin{center}
\includegraphics[scale=0.3]{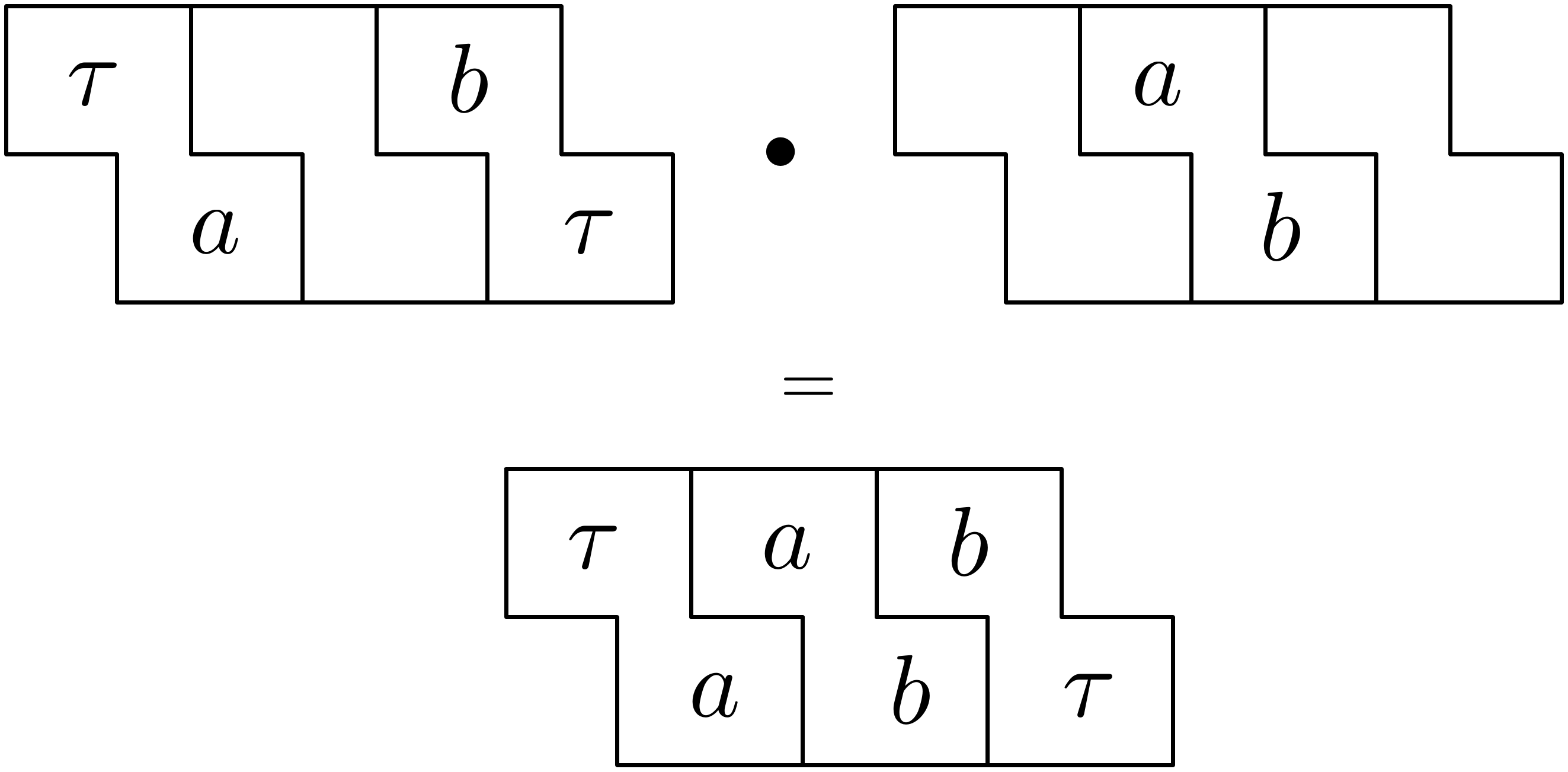}
\end{center}
\caption{Merge as assembling tetrominos.}
\label{tetrom2}
\end{figure}
 
\begin{example}
Let $s_{1}=\startchain{\silent}\chainedlink{a}{\noact}\chainedlink{\noact}{\noact}\chainend{\noact}$, 
$s_{2}=\startchain{\noact}\chainedlink{\noact}{a}\chainedlink{b}{\noact}\chainend{\noact}$, and
$s_{3}=\startchain{\noact}\chainedlink{\noact}{\noact}\chainedlink{\noact}{b}\chainend{\silent}$ 
be three link chains of the same length $|s_1|=|s_2|=|s_3|=3$.
Then $s_1$ and $s_2$ can be merged to obtain 
$s = \merges{s_{1}}{s_{2}} = (\merges{\link{\silent}{a}}{\link{\noact}{\noact}}) (\merges{\link{\noact}{\noact}}{\link{a}{b}})
(\merges{\link{\noact}{\noact}}{\link{\noact}{\noact}}) = $
$(\link{\merges{\silent}{\noact}}{\merges{a}{\noact}}) (\link{\merges{\noact}{a}}{\merges{\noact}{b}})
(\link{\merges{\noact}{\noact}}{\merges{\noact}{\noact}}) =$
$\startchain{\silent}\chainedlink{a}{a}\chainedlink{b}{\noact}\chainend{\noact}$. 
Similarly,
$s$ and $s_3$ can then be merged to obtain: 
$\merges{s}{s_{3}} = 
\startchain{\silent}\chainedlink{a}{a}\chainedlink{b}{b}\chainend{\silent}$.
\end{example}

\noindent
The merge operation enjoys some simple algebraic properties.

\begin{restatable}[]{lemma}{lemmaundici}\label{lemma:solid} 
For any $\ell,\ell',s,s'$:
\begin{enumerate}
\item[(i)]
The merge of links and link chains is commutative and associative.
\item[(ii)]
$\merges{\ell}{\ell'}= \link{\noact}{\noact}$ if and only if $\ell = \ell' = \link{\noact}{\noact}$.
\item[(iii)]
If $s$ is solid, then for any $s'$ we have $\merges{s}{s'}=\bot$.
\end{enumerate}
\end{restatable}

Finally, the following lemma about the composition of link chains will be exploited in Lemma~\ref{lem:stretchlts} to prove that the operational semantics of \CNA\ is insensitive w.r.t. the equivalence $\blackstretcheq$.

\begin{restatable}[]{lemma}{lemmadodici}\label{lem:mergestretch} 
Let $s$, $s'$, and $s''$ be three link chains such that $(\merges{s'}{s''})\  \blackstretcheq\  s$, then there must exist $s_1$ and $s_2$ such that $s_1  \blackstretcheq s'$ and $s_2 \blackstretcheq s''$
with $\merges{s_1}{s_2} = s$.
\label{lemma:eqlc}
\end{restatable}

\paragraph{Restriction}
Certain actions of the link chain can be hidden by restricting the channel where they take place. Of course, restriction of $a$ is possible only if no pending communication on $a$ (like
 $\startchain{\silent}\chainedlink{a}{\noact}\chainend{\noact}$) is present,  i.e.\  only matched communication pairs, 
 in intermediate positions,
 can be restricted
 (as in $\startchain{\silent}\chainedlink{a}{a}\chainend{\silent}$). 
 
 \begin{definition}[Matched Action]
Let $s = \ell_{1}...\ell_{n}$, with $\ell_{i} = \link{\alpha_{i}}{\beta_{i}}$ for $i\in [1,n]$.
We say that an action $a$ is \emph{matched} in $s$ if: 
\begin{enumerate}
\item $a \neq \alpha_{1},\beta_{n}$, and
\item for any $i\in [1,n-1]$, either $\beta_{i}=\alpha_{i+1}=a$ or $\beta_{i},\alpha_{i+1}\neq a$.
\end{enumerate}
Otherwise, we say that $a$ is \emph{unmatched} (or \emph{pending}) in $s$.
\end{definition}

It follows from the definition that we say that $a$ is matched in $s$ also when $a$ does not appear at all in $s$.

For instance, 
$a$ is matched in the sequence $\startchain{\silent}\chainedlink{a}{a}\chainend{\silent}$, while it is 
 pending in the sequences $\startchain{\silent}\chainedlink{a}{\noact}\chainend{\noact}$ and in $\startchain{a}\chainedlink{a}{a}\chainedlink{a}{a}\chainend{a}$.

\begin{definition}[Restriction]
Let $s = \ell_{1}...\ell_{n}$, with $\ell_{i} = \link{\alpha_{i}}{\beta_{i}}$ with $i\in [1,n]$. 
We define the restriction operation $\restrict{a} s$ by letting
\[
\begin{array}{lcllcl lcl}
\restrict{a} s & \defeq &
\left\{\begin{array}{ll}
 (\restrict{a}\ell_1)\dots (\restrict{a}\ell_n)&  \mbox{ if } a \mbox{ is \emph{matched} in } s \\
 \bot &  \mbox{ otherwise}
\end{array}\right.
\end{array}
\]
where:
\[
\begin{array}{lcl@{\hspace{4em}} lcl}
\restrict{a}\link{\alpha}{\beta} &\defeq&\link{(\restrict{a}\alpha)}{(\restrict{a}\beta)}
&
\restrict{a}\alpha&\defeq& \left\{\begin{array}{ll} \silent & \mbox{if } \alpha =a\\
\alpha & \mbox{otherwise}
\end{array}\right.
\end{array}
\]
\end{definition}

%

Restriction on links enjoy properties similar to the usual structural congruence laws for processes.

\begin{restatable}[]{lemma}{lemmaquindici}\label{lemma:res_noact} 
%
%
For any $a,b,\ell,s,s'$
\begin{enumerate}
\item[(i)]
$\restrict{a} {\ell} = \link{\noact}{\noact}$ if and only if $\ell = \link{\noact}{\noact}$.
\item[(ii)]
$\restrict{a} (\merges{s}{s'}) = \merges{s}{\restrict{a} s'}$ if $a$ does not occur in $s$.
\item[(iii)]
$\restrict{a}\restrict{b} s = \restrict{b}\restrict{a} s$.
\end{enumerate}
\end{restatable}
%
%
\begin{example}
Let $s = \startchain{\silent}\chainedlink{a}{a}\chainedlink{b}{\noact}\chainend{\noact}$ and 
$s'=\startchain{\noact}\chainedlink{\noact}{\noact}\chainedlink{\noact}{b}\chainend{\silent}$. Then, we have that
$\restrict{a} s =(\restrict{a} \link{\silent}{a}) ( \restrict{a}\link{a}{b})(\restrict{a}\link{\noact}{\noact})$ $=$
$\startchain{\silent}\chainedlink{\silent}{\silent}\chainedlink{b}{\noact}\chainend{\noact}$, while
$\restrict{a} (\merges{s}{s'}) = \startchain{\silent}\chainedlink{\silent}{\silent}\chainedlink{b}{b}\chainend{\silent} = \merges{(\restrict{a}s)}{s'}$, because $a$ does not occur in $s'$.
\end{example}

Finally, we prove a technical lemma, similar to Lemma~\ref{lem:mergestretch} for the merge, that will be exploited in the proof of Lemma~\ref{lem:stretchlts}.

\begin{restatable}[]{lemma}{lemmadiciassette}\label{lemma:res} 
Let  $s$ and $s'$ be two link chains such that  $\restrict{a}s$ is defined and $\restrict{a}s \blackstretcheq s'$, then there exists $s''$ such that $s' = \restrict{a}s''$ and  $s \blackstretcheq s''$. 
\end{restatable}

\paragraph{Renaming}

The last operation that we present is called \emph{renaming} and allows us to change, in a uniform way, the channel names appearing in a link chain. A \emph{channel renaming function} is a bijection $\phi: \mathcal{A} \to \mathcal{A}$ such that $\phi(\silent) = \silent$ and $\phi({\noact}) = \  \noact$.

\begin{definition}[Renaming]\label{def:chainren}
Let $s = \ell_{1}...\ell_{n}$, with $\ell_{i} = \link{\alpha_{i}}{\beta_{i}}$ with $i\in [1,n]$, and $\phi$ be a channel renaming function.  We define the renaming operation $s[\phi]$ by letting
\[
s[\phi] \defeq (\ell_1[\phi])\dots (\ell_n[\phi])
\qquad
(\link{\alpha}{\beta})[\phi] \defeq \link{\phi(\alpha)}{\phi(\beta)}
\]
\end{definition}

It can be readily checked that channel renaming functions enjoy the following properties.

\begin{restatable}[]{lemma}{lemmadiciannove}\label{lemma:rename} 
For any $a,\phi,\psi,\ell,s,s'$
\begin{enumerate}
\item[(i)]
$\ell[\phi] = \link{\noact}{\noact}$ if and only if $\ell = \link{\noact}{\noact}$.
\item[(ii)]
$(\merges{s}{s'})[\phi] = \merges{s[\phi]}{(s'[\phi])}$.
\item[(iii)]
$(\restrict{a} s)[\phi] = \restrict{\phi(a)} (s[\phi])$.
\item[(iv)]
$s[\phi][\psi] = s[\psi\circ \phi]$.
\item[(v)]
If $s \blackstretcheq s'$ then $s[\phi] \blackstretcheq s'[\phi]$.
\end{enumerate}
\end{restatable}

We conclude this section by proving a last technical lemma that will be exploited in the proof of Lemma~\ref{lem:stretchlts}.

\begin{restatable}[]{lemma}{lemmaventi}\label{lemma:ren} 
Let  $\phi$ be a channel renaming function and $s,s'$ be two link chains such that $s[\phi] \blackstretcheq s'$, then there exists $s''$ such that $s' = s''[\phi]$ and  
$s \blackstretcheq s''$. 
\end{restatable}


\section{A Core Network Algebra}
\label{sec:network}

In this section we build on the theory of links and link chains to present the syntax and operational semantics of \CNA.

\subsection{Syntax}

\begin{definition}
The \CNA\ processes are generated by the following grammar:
\[
\begin{array}{rcl}
P,Q & ::= &
\nil \prodsep
\ell. P  \prodsep
P+Q \prodsep
P|Q \prodsep
\restrict{a} P \prodsep
P[\phi] \prodsep
A
\end{array}
\]
\noindent
where $\ell$ is a solid link (i.e.\  $\ell = \link{\alpha}{\beta}$ with $\alpha,\beta\neq \ \noacts$),
$\phi$ is a channel renaming function,
and $A$ is a process identifier for which we assume a definition $A \defeq P$ is available in a given set $\Delta$ of (possibly recursive) process definitions.
\end{definition}

As usual, we write $\tilde{a}$ for tuples of channels and we allow parametric process definitions of the form $A(\tilde{a}) \defeq P$,
where $\tilde{a}$ is the set of free channels of $P$.
For brevity, in the examples, we sometimes write $A \defeq P$ leaving implicit that the free channles of $P$ are the parameters of $A$.

\begin{remark}
The extension in which generic link chains are allowed as action prefixes instead of solid links is discussed in Section~\ref{subsec:altdef}.
\end{remark}

It is evident that processes are built over a CCS-like syntax,
with inactive process $\nil$, action prefix $\ell.P$, choice $P+Q$, parallel $P|Q$, restriction $\restrict{a} P$, renaming $P[\phi]$ and constant definition
 $A$, but where the underlying synchronisation algebra~\cite{DBLP:journals/tcs/Winskel84} is based on link chains.
This is made evident by the operational semantics that we present next.

As usual, $\restrict{a}P$ binds the occurrences of $a$ in $P$, the sets of free and of bound names of a process $P$ are defined in the obvious way and 
processes are taken up to alpha-conversion of bound names.
We shall sometimes omit trailing $\nil$, e.g.~by writing $\link{a}{b}$ instead of $\link{a}{b}.\nil$.

\begin{example}
\CNA\ provides us with a natural way to rephrase the communication primitives of usual process calculi, such as CCS and CSP~\cite{Hoare85}, in terms of links.
\begin{itemize}
\item
Intuitively, the output action $\overline{a}$ (resp. the input action $a$) of CCS can be seen as the link $\link{\silent}{a}$
(resp. $\link{a}{\silent}$) and the solid link chain $\startchain{\silent}\chainedlink{a}{a}\chainend{\silent}$  as a dyadic communication, analogous to the silent action $\tau$ of CCS.
\item
The action $a$ of CSP can be seen as the link $\link{a}{a}$ and the solid link chain $\startchain{a}\chainedlink{a}{a}\chainedlink{a}{a}\chainend{a}$ as a CSP-like communication among three peers over $a$.
\end{itemize}
\end{example}

\subsection{Operational Semantics}

The idea is that communication can be routed across several processes by combining the links they make available to form a link chain. Since the length of the link chain is not fixed {\em a priori}, an open multi-party synchronisation is realised.



\begin{figure}[t]
\begin{center} 
\begin{prooftree} 
\AxiomC{$s \blackstretcheq \ell $} 
\RightLabel{\scriptsize(\textit{Act})}
\UnaryInfC{$\ell.P \xrightarrow{s} P$} 
\DisplayProof
\
\AxiomC{$P \xrightarrow{s} P'$} 
\RightLabel{\scriptsize(\textit{Lsum})}
\UnaryInfC{$P+Q \xrightarrow{s} P'$} 
\DisplayProof
\
\AxiomC{$P \xrightarrow{s} P'$} 
\RightLabel{\scriptsize(\textit{Res})}
\UnaryInfC{$\restrict{a}P \xrightarrow{\restrict{a}s} \restrict{a}P'$} 
\end{prooftree} 
\end{center}

\begin{center} 
\begin{prooftree} 
\AxiomC{$P \xrightarrow{s} P'$} 
\RightLabel{\scriptsize(\textit{Ren})}
\UnaryInfC{$P[\phi] \xrightarrow{s[\phi]} P'[\phi]$} 
\DisplayProof
\
\AxiomC{$P \xrightarrow{s} P'$} 
\RightLabel{\scriptsize(\textit{Lpar})}
\UnaryInfC{$P|Q \xrightarrow{s} P'|Q$} 
\DisplayProof
\ 
\AxiomC{$P \xrightarrow{s} P'$} 
\AxiomC{$Q \xrightarrow{s'} Q'$} 
\RightLabel{\scriptsize(\textit{Com})}
\BinaryInfC{$P|Q \xrightarrow{\merges{s}{s'}} P'|Q'$} 
\DisplayProof
\end{center}
\begin{center}
\AxiomC{$P \xrightarrow{s} P'$}
 \AxiomC{$(A \defeq P)\in\Delta$}
\RightLabel{\scriptsize(\textit{Ide})}
\BinaryInfC{$A \xrightarrow{s} P'$} 
\end{prooftree} 
\end{center}
\caption{SOS semantics of the \CNA\  (rules (\textit{Rsum}) and (\textit{Rpar}) are omitted).}
\label{fig:cnasos}
\end{figure}

The operational semantics is defined in terms of a Labelled Transition System, in which states are \CNA\ processes, labels are link chains, and transitions are generated by the SOS rules in Figure~\ref{fig:cnasos}. 
%
Notice that the rules are very similar to the ones of CCS, 
apart from the labels that record the link chains involved in the transitions:
moving from dyadic to \emph{linked} interaction does not introduce any complexity burden from the formal point of view.

We comment in details the rules (\textit{Act}), (\textit{Res}), and (\textit{Com}). 
In rules (\textit{Res}) and (\textit{Com}) we leave implicit the side conditions $\restrict{a}s \neq \bot$ and $\merges{s}{s'} \neq \bot$, respectively (they can be easily recovered by noting that otherwise the label of the transition in the conclusion would be undefined).

The rule (\textit{Act}) states that $\ell.P \xrightarrow{s} P$ for any link chain $s$, whose unique solid link is $\ell$,  i.e.\  any $s$ such that $s \blackstretcheq \ell$
(we recall that $s \blackstretcheq \ell$ if $s$ and $\ell$  differ only for the presence of virtual links).
Intuitively, $\ell.P$ can take part in any interaction, in any (admissible) position.
To join in a communication, $\ell.P$ should exhibit the capability to enlarge its link $\ell$ to a link chain $s \blackstretcheq \ell$, whose length is the same as the length of the chains offered by all the other participants, so to proceed with the merge operation. Following the early style, the suitable length is 
inferred at the time of deducing the input transition.
Note that, by definition of link chain, if one site of $\ell$ is $\silent$, then $\ell$ can only appear at one of the extremes of $s$.

The rule (\textit{Res}) can serve different aims: 
(i) \emph{floating}, if $a$ does not occur in $s$, then $\restrict{a}s = s$ and  $\restrict{a}P \xrightarrow{s} \restrict{a}P'$;
(ii) \emph{hiding}, if $a$ is matched in $s$ ( i.e.\  $a$ appears as sites already matched by adjacent links), then all occurrences of $a$ in $s$ are transformed to $\silent$ in $\restrict{a}s$;
(iii) \emph{blocking}, if $a$ is pending in $s$ ( i.e.\  there are some
unmatched occurrences of $a$ in $s$), then $\restrict{a}s = \bot$ and the rule cannot be applied.

In the (\textit{Com}) rule the link chains recorded on both the premises' transitions are merged in the conclusion's transition.
This is possible only if $s$ and $s'$ are to some extent ``complementary''.
Contrary to CCS, the rule (\textit{Com}) can appear several times in the proof tree of a transition, because $\merges{s}{s'}$ can still contain virtual links (if $s$ and $s'$ had a virtual link in the same position) and can possibly be merged with other link chains.
However, when $\merges{s}{s'}$ is solid, no further synchronisation is possible (by Lemma~\ref{lemma:solid} (ii)).

%
%
%
\begin{example}
Let
$P = \link{\silent}{a}.P_{1}\ |\ \restrict{b} Q$ and $Q = \link{b}{\silent}.P_{2}\ |\ \link{a}{b}.\nil$, for some processes $P_1$ and $P_2$.
The process $\link{\silent}{a}.P_{1}$ can perform an output on $a$, the process $ \link{b}{\silent}.P_{2}$ can perform an input from $b$; the process $\link{a}{b}$ provides a one-shot link forwarder from $a$ to $b$. 
Their links match along the sequence and a three-party interaction can take place.
Together, these processes can indeed synchronise by agreeing to form the solid link chain
${\startchain{\silent}\chainedlink{a}{a}\chainedlink{\silent}{\silent}\chainend{\silent}}$ of length three, as follows.\footnote{Note that $b$ is restricted and therefore the matched communication on $b$ is replaced by $\tau$ in the observed chain.}

{\small
\begin{center} 
\begin{prooftree} 
\AxiomC{} 
\RightLabel{\scriptsize(\textit{Act})}
\UnaryInfC{$\link{\silent}{a}.P_{1} \xrightarrow{\startchain{\silent}\chainedlink{a}{\noact}\chainedlink{\noact}{\noact}\chainend{\noact}} P_{1}$} 
\! \!\!\! \!\!
      \AxiomC{} 
      \RightLabel{\scriptsize(\textit{Act})}
      \UnaryInfC{$\link{b}{\silent}.P_{2} \xrightarrow{\startchain{\noact}\chainedlink{\noact}{\noact}\chainedlink{\noact}{b}\chainend{\silent}} P_{2}$} 
      \AxiomC{} 
      \RightLabel{\scriptsize(\textit{Act})}
      \UnaryInfC{$\link{a}{b}.\nil \xrightarrow{\startchain{\noact}\chainedlink{\noact}{a}\chainedlink{b}{\noact}\chainend{\noact}} \nil$} 
   \RightLabel{\scriptsize(\textit{Com})}
\BinaryInfC{$Q \xrightarrow{\startchain{\noact}\chainedlink{\noact}{a}\chainedlink{b}{b}\chainend{\silent}} P_{2}|\nil$} 
%
   \RightLabel{\scriptsize(\textit{Res})}
   \UnaryInfC{$\restrict{b}Q \xrightarrow{\startchain{\noact}\chainedlink{\noact}{a}\chainedlink{\silent}{\silent}\chainend{\silent}} \restrict{b} (P_{2}|\nil)$} 
\RightLabel{\scriptsize(\textit{Com})}
\BinaryInfC{$P \xrightarrow{\startchain{\silent}\chainedlink{a}{a}\chainedlink{\silent}{\silent}\chainend{\silent}} P_{1}|\restrict{b} (P_{2}|\nil)$} 
\end{prooftree} 
\end{center}
}


\end{example}

The following lemma, whose proof goes by rule induction, shows that labels behave like an accordion.
Concretely, any label $s$ in a transition is replaceable with any other chain having a different number of virtual links $\link{\noact}{\noact}$ added to $s$ according to the axioms of $\blackstretcheq$. 
The result builds on the previous technical Lemmata~\ref{lem:mergestretch} and~\ref{lemma:res}.
This fact will be later exploited in Section~\ref{abs-sem}, where the abstract semantics is given.

\begin{lemma}[Accordion Lemma]\label{lem:stretchlts}
If $P\xrightarrow{s} P'$ and $s \blackstretcheq s'$,
then $P\xrightarrow{s'} P'$.
\end{lemma}
\begin{proof}
The proof is by rule induction.
\begin{description}
\item[{\bf rule (Act)}]
Assume $\ell.P\xrightarrow{s} P$ with $s \blackstretcheq s'$. We need to prove that $\ell.P\xrightarrow{s'} P$
Since by hypothesis $s \blackstretcheq \ell$ then, by transitivity,  $s' \blackstretcheq \ell$.
By applying rule (Act) we get  the thesis $\ell.P \xrightarrow{s'}P$.
\item[{\bf rule (Lsum)}]
Assume $P+Q \xrightarrow{s} P'$ with $P \xrightarrow{s}P'$ and  $s \blackstretcheq s'$.  
We need to prove that $P+Q \xrightarrow{s'} P'$. By inductive hypothesis 
we have  that $P \xrightarrow{s'}P'$, and by applying rule $(Lsum)$ we get the thesis.
For the rules {\bf (LPar)}, and {\bf (Ide)} the proof is similar to this case and thus omitted.
\item[{\bf rule (Res)}]
Assume $\restrict{a}P  \xrightarrow{\restrict{a}s} \restrict{a}P'$ with $P  \xrightarrow{s}P'$ and $\restrict{a}s\blackstretcheq s'$. We want to prove that $\restrict{a}P  \xrightarrow{s'} \restrict{a}P'$.
By Lemma~\ref{lemma:res}, there exists $s''$ s.t. $s' = \restrict{a}s''$ with $s\blackstretcheq s''$. Then, by inductive hypothesis we have that $P  \xrightarrow{s''}P'$, and by applying rule $(Res)$ we obtain the thesis.
\item[{\bf rule (Ren)}] 
Assume $P[\phi]  \xrightarrow{s[\phi]} P'[\phi]$ with $P  \xrightarrow{s}P'$ and $s[\phi] \blackstretcheq s'$. We want to prove that $P[\phi]  \xrightarrow{s'} P'[\phi]$. By Lemma~\ref{lemma:ren}, there exists $s''$ such that $s' = s''[\phi]$ and $s \blackstretcheq s''$. Then, by inductive hypothesis we have that $P  \xrightarrow{s''}P'$, and by applying rule 
$(Ren)$ we get the thesis.
\item[{\bf rule (Com)}]
Assume $P|Q \xrightarrow{s} P'|Q'$ and $s \blackstretcheq s'$. We want to prove that $P|Q \xrightarrow{s'} P'|Q'$.
By hypothesis, there exist $s_1$ and $s_2$ such that $P \xrightarrow{s_1} P'$ and 
$Q \xrightarrow{s_2} Q'$, with $s = \merges{s_1}{s_2}$.
By Lemma~\ref{lemma:eqlc}, there exist $s_1'$ and $s_2'$ such that $s_1' \blackstretcheq s_1$, 
$s_2' \blackstretcheq s_2$, and  $\merges{s_1'}{s_2'} = s'$. By inductive hypothesis, 
 $P \xrightarrow{s_1'} P'$ and  $Q \xrightarrow{s_2'} Q'$. Finally, by applying 
 rule $(Com)$ we get the thesis.
\end{description}
\end{proof}

\begin{example}[Forwarders]
We give a few examples to show how flexible is \CNA\ for defining ``forwarding'' policies.
We have already seen a one-shot and one-hop forwarder from $a$ to $b$ that can be written as $\link{a}{b}.\nil$.
Its persistent version is just written as $R(a,b) \defeq \link{a}{b}.R(a,b)$.
Moreover, the process $R(a,b) \ | R(b,a)$ provides a sort of name fusion, making $a$ and $b$ interchangeable.

An alternating forwarder $A(a,b,c)$ from $a$ to $b$ first and then to $c$ can be defined as 
\[A(a,b,c) \defeq \link{a}{b}.\link{a}{c}.A(a,b,c).\]

A persistent non-deterministic forwarder $C(a,\widetilde{c})$ (the $C$ stands for \emph{choice}), from $a$ to $c_{1},...,c_{n}$ can be \mbox{written as} 
\[C(a,\widetilde{c}) \defeq \link{a}{c_{1}}.C(a,\widetilde{c})+ \cdots + \link{a}{c_{n}}.C(a,\widetilde{c}).\]

Similarly, $J(\widetilde{b},a)$ defined as (the $J$ stands for \emph{join}\footnote{With the term ``join'' we refer to the fact that messages from different sources are forwarded to the same channel. It has no relation with join patterns in join calculus.})
\[J(\widetilde{b},a)\defeq \link{b_{1}}{a}.J(\widetilde{b},a) + \cdots + \link{b_{m}}{a}.J(\widetilde{b},a)\] 
is a persistent non-deterministic forwarder, from $b_{1},...,b_{m}$ to $a$.

By combining the two processes above as in
\[F(\widetilde{b},\widetilde{c}) \defeq \restrict{a}(J(\widetilde{b},a)\ |\ C(a,\widetilde{c}))\]
we obtain a persistent forwarder from any of  the $b_{i}$'s  to any of the $c_{j}$'s.
The only admissible transitions have indeed the form
$
F(\widetilde{b},\widetilde{c})  \xrightarrow{s} F(\widetilde{b},\widetilde{c})
$
with $s \blackstretcheq \startchain{b_i}\chainedlink{\silent}{\silent}\chainend{c_j}$ for some $i\in[1,m]$ and $j\in[1,n]$ (because interaction on $a$ is restricted and $\restrict{a}(\startchain{b_i}\chainedlink{a}{a}\chainend{c_j}) = \startchain{b_i}\chainedlink{\silent}{\silent}\chainend{c_j}$).
\end{example}

\begin{example}[Blind routing in \CNA]
We have already observed that CCS processes can be transformed to \CNA\ processes just by transforming action prefixes in link prefixes. Correspondingly, the blind routing example from Section~\ref{sec:CCS} (Example~\ref{blind}) can be encoded in \CNA\ just as follows, where for readability we have omitted the parameters from definitions.
\begin{eqnarray*}
A_i & \defeq & \link{\tau}{\mathit{req}_i}.\link{\tau}{\mathit{think}}.A_i\quad \mbox{ for $i\in[1,2]$} \\
S_j & \defeq & \link{\mathit{srv}_j}{\tau}.\link{\tau}{\mathit{exec}}.S_j + \link{\tau}{\mathit{busy}}.\link{\tau}{\tau}.S_j\quad \mbox{ for $j\in[1,2]$}\\
R & \defeq & \link{\mathit{req}_1}{\tau}.(\link{\tau}{\mathit{srv}_1}.R + \link{\tau}{\mathit{srv}_2}.R) + \link{\mathit{req}_2}{\tau}.\link{\tau}{\mathit{srv}_2}.R \\
N & \defeq &
\restrict{\widetilde{\mathit{req}}} 
\restrict{\widetilde{\mathit{srv}}}
(A_1\ |\ A_2\ |\ R\ |\ S_1\ |\ S_2)
\end{eqnarray*}

For example, the following transitions can be derived from the SOS rules accounting for the case where a request from $A_1$ is assigned to $S_2$:
\begin{eqnarray*}
N
& \xrightarrow{ \startchain{\tau}\chainedlink{\tau}{\tau}\chainend{\tau} } &
\restrict{\widetilde{\mathit{req}}} 
\restrict{\widetilde{\mathit{srv}}}
(\link{\tau}{\mathit{think}}.A_1\ |\ A_2\ |\ (\link{\tau}{\mathit{srv}_1}.R + \link{\tau}{\mathit{srv}_2}.R)\ |\ S_1\ |\ S_2)
\\
& \xrightarrow{ \startchain{\tau}\chainedlink{\tau}{\tau}\chainend{\tau} } &
\restrict{\widetilde{\mathit{req}}} 
\restrict{\widetilde{\mathit{srv}}}
(\link{\tau}{\mathit{think}}.A_1\ |\ A_2\ |\ R\ |\ S_1\ |\ \link{\tau}{\mathit{exec}}.S_2)
\\
& \xrightarrow{ \link{\tau}{\mathit{think}} } &
\restrict{\widetilde{\mathit{req}}} 
\restrict{\widetilde{\mathit{srv}}}
(A_1\ |\ A_2\ |\ R\ |\ S_1\ |\ \link{\tau}{\mathit{exec}}.S_2)
\\
& \xrightarrow{ \link{\tau}{\mathit{exec}} } &
N
\end{eqnarray*}

Note that, in the first two steps, interactions on channels $\mathit{req}_1$ and $\mathit{srv}_2$ are restricted and thus not observable, in fact $\restrict{\mathit{req}_1}(\startchain{\tau}\chainedlink{\mathit{req}_1}{\mathit{req}_1}\chainend{\tau}) = \startchain{\tau}\chainedlink{\tau}{\tau}\chainend{\tau}$ and $\restrict{\mathit{srv}_2}(\startchain{\tau}\chainedlink{\mathit{srv}_2}{\mathit{srv}_2}\chainend{\tau}) = \startchain{\tau}\chainedlink{\tau}{\tau}\chainend{\tau}$.

\end{example}

\begin{example}[Acknowledged routing in \CNA]
The features of \CNA\ become evident in our running example when we come to modelling acknowledged routing  (Example~\ref{ack}). This is because the explicit acknowledgment is no longer necessary as it can be implicitly accounted for by the ability to construct a chain of links.\footnote{Of course, one could also just encode the CCS processes for acknowledged routing just by translating their prefixes as we have done for blind routing.}
Correspondingly, we set
\begin{eqnarray*}
A_i & \defeq & \link{\tau}{\mathit{req}_i}.\link{\tau}{\mathit{think}}.A_i\quad \mbox{ for $i\in[1,2]$} \\
S_j & \defeq & \link{\mathit{srv}_j}{\tau}.\link{\tau}{\mathit{exec}}.S_j + \link{\tau}{\mathit{busy}}.\link{\tau}{\tau}.S_j\quad \mbox{ for $j\in[1,2]$}\\
R & \defeq & \link{\mathit{req}_1}{\mathit{srv}_1}.R + \link{\mathit{req}_1}{\mathit{srv}_2}.R + \link{\mathit{req}_2}{\mathit{srv}_2}.R \\
M & \defeq &
\restrict{\widetilde{\mathit{req}}} 
\restrict{\widetilde{\mathit{srv}}}
(A_1\ |\ A_2\ |\ R\ |\ S_1\ |\ S_2)
\end{eqnarray*}
Note that channels $\mathit{ack}_i$ are not needed and \CNA\ processes $A_i$ and $S_j$ are defined as in the case of blind routing; only the syntax of $R$ has been changed to link, in one single step, the requests from agents with the availability of servers.

For example, the following transitions can be derived from the SOS rules accounting for the case where a request from $A_1$ is assigned to $S_2$:
\begin{eqnarray*}
M
& \xrightarrow{ \startchain{\tau}\chainedlink{\tau}{\tau}\chainedlink{\tau}{\tau}\chainend{\tau} } &
\restrict{\widetilde{\mathit{req}}} 
\restrict{\widetilde{\mathit{srv}}}
(\link{\tau}{\mathit{think}}.A_1\ |\ A_2\ |\ R\ |\ S_1\ |\ \link{\tau}{\mathit{exec}}.S_2)
\\
& \xrightarrow{ \link{\tau}{\mathit{think}} } &
\restrict{\widetilde{\mathit{req}}} 
\restrict{\widetilde{\mathit{srv}}}
(A_1\ |\ A_2\ |\ R\ |\ S_1\ |\ \link{\tau}{\mathit{exec}}.S_2)
\\
& \xrightarrow{ \link{\tau}{\mathit{exec}} } &
M
\end{eqnarray*}

Note that, in the first step, the interaction on channels $\mathit{req}_1$ and $\mathit{srv}_2$ is restricted and thus not observable, in fact $\restrict{\mathit{req}_1}\restrict{\mathit{srv}_2}(\startchain{\tau}\chainedlink{\mathit{req}_1}{\mathit{req}_1}\chainedlink{\mathit{srv}_2}{\mathit{srv}_2}\chainend{\tau}) = \startchain{\tau}\chainedlink{\tau}{\tau}\chainedlink{\tau}{\tau}\chainend{\tau}$.

\end{example}

\begin{example}[Composite, acknowledged routing in \CNA]
Notably, the infrastructure presented in the previous example is already designed in a modular way: infrastructures can be composed without requiring any change and no dead path detection or backtracking mechanisms have to be put in place, because they are taken care by the  \CNA\ ``middleware''.

For instance, take the \CNA\ versions of the infrastructures presented in Example~\ref{composite}:
\begin{eqnarray*}
R' & \defeq & \link{\mathit{req}_1}{s_1}.R' + \link{\mathit{req}_1}{s_2}.R' + \link{\mathit{req}_2}{s_2}.R' \\
R'' & \defeq & \link{s_1}{s'_1}.R'' + \link{s_2}{s'_2}.R' \\
R''' & \defeq & \link{s'_2}{\mathit{srv}_2}.R''' \\
R & \defeq & \restrict{\widetilde{s'}}\restrict{\widetilde{s}} (R'\ |\ R''\ |\ R''')
\end{eqnarray*}

Besides looking considerably more concise than their CCS versions, they make evident that the delivery of any request to a server is atomic as any process executes one (link) action and reduces (recursively) to itself.

At a closer inspection, one may notice that the only possible transitions for $R$ are the following ones 
(up to $\blackstretcheq$, as explained by the Accordion Lemma~\ref{lem:stretchlts}):
\[
R \xrightarrow{\startchain{\mathit{req}_1}\chainedlink{\tau}{\tau}\chainedlink{\tau}{\tau}\chainend{\mathit{srv}_2}} R
\quad\mbox{ and }\quad
R \xrightarrow{\startchain{\mathit{req}_2}\chainedlink{\tau}{\tau}\chainedlink{\tau}{\tau}\chainend{\mathit{srv}_2}} R.
\]

As in the previous examples, note that, e.g.\ in the first step, the interaction on channels $s_2$ and $s'_2$ is restricted and thus not observable, in fact
$\restrict{s_2}\restrict{s'_2} \startchain{\mathit{req}_1}\chainedlink{s_2}{s_2}\chainedlink{s'_2}{s'_2}\chainend{\mathit{srv}_2}= 
\startchain{\mathit{req}_1}\chainedlink{\tau}{\tau}\chainedlink{\tau}{\tau}\chainend{\mathit{srv}_2}$.

Therefore, at a suitable level of abstraction, in which routing details are omitted, we would like to relate the composite infrastructure $R$ with the monolithic infrastructure $R_m$ defined as follows:
\begin{eqnarray*}
R_m & \defeq & \link{\mathit{req}_1}{\mathit{srv}_2}.R_m + \link{\mathit{req}_2}{\mathit{srv}_2}.R_m
\end{eqnarray*}

In the next section we show how this can be formalised.
\end{example}


\section{Abstract semantics}\label{abs-sem}


As usual, we can use the LTS semantics to define suitable behavioural equivalences over processes.
We focus on bisimulation relations.
The accordion lemma (Lemma~\ref{lem:stretchlts}) implies that it makes no sense to distinguish between two labels $s$ and $s'$ such that $s \blackstretcheq s'$, because if one process $p$ can do $s$ and reach $p'$, then it can also do $s'$ and still reach $p'$. 
However, when comparing two labels we would like to abstract away also from the number of hops performed and from the size (not just the length) of the chains, as the following example suggests.

\begin{example}\label{ex:motivatewhitetie}
Take the one-hop forwarder $R(a,b) \defeq \link{a}{b}.R(a,b)$. 
From the operational semantics it is immediate to check that its transitions are all and only
$R(a,b) \xrightarrow{s} R(a,b)$ such that $s \blackstretcheq \link{a}{b}$.

Now connect together two one-hop forwarders in a sequence to form the routing infrastructure 
$T(a,b) \defeq \restrict{c}(R(a,c)\ |\ R(c,b))$. Again it is immediate to check that its transitions are all and only
$T(a,b) \xrightarrow{s} T(a,b)$ such that $s \blackstretcheq \startchain{a}\chainedlink{\silent}{\silent}\chainend{b}$.

Intuitively, we would like $R(a,b)$ and $T(a,b)$ to be interchangeable, as they offer the same routing service.
If we were to compare $R(a,b)$ and $T(a,b)$ using plain bisimilarity, where labels must be matched syntactically, then they would not be equivalent. 
Also if we relax bisimulation by matching transition labels up-to $\blackstretcheq$, the two terms are not equivalent, as it is not true that $\link{a}{b} \blackstretcheq \startchain{a}\chainedlink{\silent}{\silent}\chainend{b}$ (they have different sizes,
i.e.\  they have a different number of solid links,
 and size is preserved by $\blackstretcheq$).
\end{example}

To define a bisimilarity equivalence that relates processes such as $R(a,b)$ and $T(a,b)$ above, we introduce an equivalence coarser than $\blackstretcheq$, written $\stretcheq$, that we use to match labels in the bisimulation game, according to which e.g.\
$\link{a}{b} \stretcheq \startchain{a}\chainedlink{\silent}{\silent}\chainend{b}$.

\begin{definition}[Equivalence $\stretcheq$]\label{def:accordion}
We let $\stretcheq$ be the least equivalence relation 
over link chains closed under the following inference rules:
\[
\frac{s \blackstretcheq s'}{s \stretcheq s'}
\qquad\qquad
s_{1}\startchain{\alpha}\chainedlink{\silent}{\silent}\chainend{\beta}s_{2} \stretcheq s_{1}\link{\alpha}{\beta}s_{2} 
\]
\end{definition}


The only difference between $\stretcheq$ and $\blackstretcheq$ is that the additional axiom of $\stretcheq$ abstracts away from intermediate matched actions that are silent.
The intuition is that such matched actions cannot be split and used to compose longer chains, because they are silent and therefore the matching was made on a restricted channel.

\begin{remark}
We invite the reader to check that the Accordion Lemma~\ref{lem:stretchlts} is only concerned with $\blackstretcheq$ and not with $\stretcheq$.
\end{remark}

We now consider the  equivalences classes given by $\stretcheq$ and its representatives.

\begin{definition}[Essential chain]
A link chain is \emph{essential} if it is composed by alternating solid and virtual links, with solid links at its extremes.
\end{definition}

For example, the chain $\startchain{a}\chainedlink{\silent}{\silent}\chainedlink{b}{b}\chainend{c}$ is not essential, while the chain $\startchain{a}\chainedlink{b}{\noact}\chainedlink{\noact}{b}\chainend{c}$ is essential and we have 
$\startchain{a}\chainedlink{\silent}{\silent}\chainedlink{b}{b}\chainend{c}
\stretcheq
\startchain{a}\chainedlink{b}{\noact}\chainedlink{\noact}{b}\chainend{c}$.

The following lemma shows that 
each $\stretcheq$-equivalence class has a unique essential representative.

\begin{restatable}[]{lemma}{lemmatrentaquattro}\label{lemma:stretcheqone} 
All of the following properties hold for any link chain $s$.
\begin{enumerate}
\item[(i)]
There exists an essential link chain $s'$ such that $s \stretcheq s'$.
\item[(ii)]
If $s$ is essential, for any essential $s'$ such that $s \stretcheq s'$, then $s=s'$.
\end{enumerate}
\end{restatable}

It is immediate to check that by orienting the axioms in Definitions~\ref{def:black} and~\ref{def:accordion} from left to right we have a procedure for transforming any link chain $s$ to a unique essential link chain $s'$ such that $s \stretcheq s'$. We write $\reduce{s}$ to denote such a unique representative.

\begin{corollary}\label{cor:essential}
For any link chains $s,s'$ we have $s \stretcheq s'$ iff $\reduce{s} = \reduce{s'}$.
\end{corollary}


The following properties are useful in the proof of the main result,  i.e.\  the congruence property of our notion of bisimilarity (Theorem~\ref{theo:CNAcong}). 

The first lemma says that the restriction operator respects the relation $\stretcheq$, in the sense that if one link chain $s$ can be restricted in $a$ then any chain $s' \stretcheq s$ can be extended to some $s'' \blackstretcheq s'$ where $a$ is matched and can be restricted.

\begin{restatable}[]{lemma}{lemmatrentasei}\label{lem:reswhite} 
If $s \stretcheq s'$, then for any $a$ such that $\restrict{a}s\neq \bot$ there exists $s'' \blackstretcheq s'$ such that $\restrict{a}s'' \neq \bot$ and $\restrict{a}s \stretcheq \restrict{a}s''$.
\end{restatable}

The second lemma says that taken two link chains $s_{1}$ and $s'_{1}$ in the same equivalence class, and given any link chain $s_2$ that can be merged with $s_{1}$, then it is possible to find a link chain $s'_2$ (which is a stretched version of $s_2$) such that it can be merged with another link chain $s''_1$ (which is a stretched version of $s'_1$) with the result being equivalent to $\merges{s_2}{s_{1}}$. This is graphically rendered in Figure~\ref{fig:lem_mergewhite}.

\begin{figure}
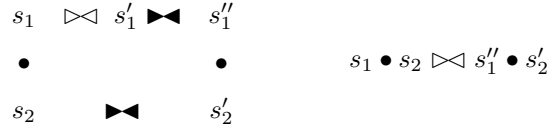

\[
\begin{array}{cc}
\begin{array}{ccc}
s_1 & \stretcheq \ s'_1  \blackstretcheq  & s''_1\\
\bullet && \bullet\\
s_2 & \blackstretcheq & s'_2
\end{array}
& \hspace{1cm}
\merges{s_1}{s_2} \stretcheq  \merges{s''_1}{s'_2}
\end{array}
\]
\caption{Graphic representation of Lemma~\ref{lem:mergewhite}.}
\label{fig:lem_mergewhite}
\end{figure}


\begin{restatable}[]{lemma}{lemmatrentasette}\label{lem:mergewhite} 
If $s_1  \stretcheq s'_1$, then for any $s_2$ such that $\merges{s_2}{s_{1}}\neq \bot$ there exist two link chains $s'_2 \blackstretcheq s_2$ and $s''_1\blackstretcheq s'_1$ such that
$\merges{s'_2}{s''_{1}}\neq \bot$ and $\merges{s_2}{s_1}  \stretcheq \merges{s'_2}{s''_1}$.
\end{restatable}

Also the next lemma is introduced to prove the main Theorem~\ref{theo:CNAcong}.

\begin{restatable}[]{lemma}{lemmatrentotto}\label{lem:reduce} 
Let $s$ and $s'$ be two link chains such that $s  \stretcheq s'$, then for any renaming function $s[\phi]  \stretcheq s'[\phi]$.
\end{restatable}



\begin{definition}[Network bisimulation]
A \emph{network bisimulation} $\mathbf{R}$ is a binary relation over CNA processes such that, if $P \mathrel{\mathbf{R}} Q$ then:
\begin{itemize}
\item if $P \xrightarrow{s} P'$, then $\exists$ $s'$, $Q'$ such that $s' \stretcheq s$, $Q  \xrightarrow{s'} Q'$, and $P' \mathrel{\mathbf{R}} Q'$;
\item if $Q \xrightarrow{s} Q'$, then $\exists$ $s'$, $P'$ such that $s' \stretcheq s$, $P  \xrightarrow{s'} P'$, and $P' \mathrel{\mathbf{R}} Q'$.
\end{itemize}
\end{definition}

Note that, by Corollary~\ref{cor:essential}, the requirement $s' \stretcheq s$ amounts just to checking that $\reduce{s'} = \reduce{s}$.

\begin{definition}[Network bisimilarity $\simnet$]
We let $\simnet$ denote the largest network bisimulation and we say that $P$ is \emph{network bisimilar} to $Q$ if $P \simnet Q$. 
\end{definition}

It can be shown that network bisimulations are closed with respect to union and composition and that $\simnet$ is an equivalence relation.

\begin{example}
Take the recursively defined processes $R(a,b) \defeq \link{a}{b}.R(a,b)$  and $T(a,b) \defeq \restrict{c}(R(a,c)\ |\ R(c,a))$ from Example~\ref{ex:motivatewhitetie}. It is straightforward to check that the relation 
\[
\mathbf{R} \defeq \{ (R(a,b),T(a,b)) \}
\]
is a network bisimulation, because $\link{a}{b}\stretcheq \startchain{a}\chainedlink{\silent}{\silent}\chainend{b}$, hence $R(a,b)$ and $T(a,b)$ are network bisimilar.
\end{example}

\begin{example}
Consider the two processes $P \defeq \link{a}{b}.P$ and  $Q \defeq \restrict{c}(\link{a}{c} \ | \ \link{c}{b}.Q)$.
We have that whenever $P  \xrightarrow{s} P'$, then $P'=P$ and $\reduce{s}= \link{a}{b}$.
Similarly, whenever $Q  \xrightarrow{s} Q'$, then $Q'=\restrict{c}(\nil | Q)$ and $\reduce{s}= \link{a}{b}$.
Then we prove that $P \simnet Q$ by showing that the relation $\mathbf{R}$ below:
\[
\mathbf{R} \defeq \{ (P,R) \mid \exists n. R = C^{n}[Q] \}
\]
\noindent
is a network bisimulation, where $C^{n}[Q]$ is inductively defined by letting $C^{0}[Q] \defeq Q$ 
and $C^{n+1}[Q] \defeq C[C^{n}[Q]]$ for  $C[\cdot]$  the context $\restrict{c}(\nil | \cdot )$.
Intuitively, the context $C^{n}$ mimics the effects of $n$ internal interactions in $Q$, as any internal interaction in $Q$ leaves a zero process in parallel with $Q$. 
The thesis simply follows by noting that, for any $n$, whenever $C^n[Q]  \xrightarrow{s} R'$ then $R' = C^{n+1}[Q]$ and $\reduce{s}= \link{a}{b}$: by induction on $n$, the base case $C^0[Q] = Q$ has been already observed above, while the inductive case, where we consider 
$C^{n+1}[Q] = \restrict{c}(\nil | C^{n}[Q] )$, follows immediately from the inductive hypothesis.
\end{example}

We are now ready to prove the first main result, i.e.\ that network bisimilarity is preserved by all the operators of \CNA.

\begin{restatable}[]{theorem}{theoremquarantatre}\label{theo:CNAcong}
Network bisimilarity is a congruence.
\end{restatable}

\begin{proof}
We show that network bisimilarity is preserved by all the operators.
The proof uses standard arguments. The interesting cases are that of restriction, renaming and parallel composition, the others are just suitable rephrasing of the corresponding proofs in the CCS literature.
\begin{description}
\item[Prefix]
We want to prove that if $P \simnet Q$ then for any $\ell$ we have $\ell.P \simnet \ell.Q$.
We define the relation 
$\mathbf{R}_{pre}\defeq \{(\ell.P,\ell.Q) \mid P \simnet Q\} \ \cup \simnet$  and show that $\mathbf{R}_{pre}$ is a network bisimulation.
The case when $(P,Q) \in \ \simnet$ is obvious.
Take $(\ell.P,\ell.Q) \in \{(\ell.P,\ell.Q) \mid P \simnet Q\}$.
If $\ell.P \xrightarrow{s} P$ with $s\blackstretcheq \ell$ then also $\ell.Q \xrightarrow{s} Q$ and 
$P \mathrel{\mathbf{R}_{pre}} Q$ as $P \simnet Q$.
Vice versa, if $\ell.Q \xrightarrow{s} Q$ with $s\blackstretcheq \ell$ then also $\ell.P \xrightarrow{s} P$ and 
$P \mathrel{\mathbf{R}_{pre}} Q$ as $P \simnet Q$.

\item[Restriction]
We want to prove that if $P \simnet Q$ then for any $a$ we have $\restrict{a}P \simnet \restrict{a}Q$.
We let $\mathbf{R}_{res}\defeq \{(\restrict{a}P,\restrict{a}Q)|  P \simnet Q\}$
and show that $\mathbf{R}_{res}$ is a network bisimulation.
Suppose $P \simnet Q$ and $\restrict{a}P \xrightarrow{\restrict{a}s} \restrict{a}P'$, for some $s$ and $P'$
such that $P \xrightarrow{s} P'$.
By assumption, we know that $P \simnet Q$ and therefore there exist $s'$, $Q'$ such that
$Q \xrightarrow{s'}Q'$ with $\reduce{s'}=\reduce{s}$ and $P' \simnet Q'$.
By Lemma~\ref{lem:reswhite}, there exists $s'' \blackstretcheq s'$ such that
 $\restrict{a}s''\neq \bot$ and $\restrict{a}s'' \stretcheq \restrict{a}s$.
Hence $\restrict{a}Q \xrightarrow{\restrict{a} s''}\restrict{a}Q'$ and,
by definition of $\mathbf{R}_{res}$, we obtain $\restrict{a}P' \mathrel{\mathbf{R}_{res}} \restrict{a}Q'$.

\item[Renaming]
We want to prove that if $P \simnet Q$ then for any renaming function $\phi$ we have $P[\phi] \simnet Q[\phi]$. Let $\mathbf{R}_{ren} \defeq \{(P[\phi],Q[\phi]) \mid P \simnet Q\}$ and show that  $\mathbf{R}_{ren}$ is a network bisimulation.
Suppose $P \simnet Q$ and $P[\phi] \xrightarrow{s}P'[\phi]$, for some $s$, $P'$.
By rule $(Ren)$, it must exist $s'$ such that $s= s'[\phi]$, and $P \xrightarrow{s'} P'$. 
By assumption we know that  $P \simnet Q$ and therefore there exist $s''$, and $Q'$ s.t. $Q \xrightarrow{s''}Q'$ with    $\reduce{s'}=\reduce{s''}$.
By Lemma~\ref{lem:reduce}, we have $\reduce{s'[\phi]}= \reduce{s''[\phi]}$. Hence
$Q[\phi] \xrightarrow{s''[\phi]} Q'[\phi]$ and, by definition of $\mathbf{R}_{ren}$, we obtain 
$P[\phi] \ \mathbf{R}_{ren} \ Q[\phi]$.


\item[Choice]
We want to prove that if $P \simnet Q$ then for any $R$ we have $P+R \simnet Q+R$.
We define the relation 
$\mathbf{R}_{sum} \defeq \{(P+R,Q+R) \mid P \simnet Q\}\cup \simnet$ 
and show that $\mathbf{R}_{sum}$ is a network bisimulation.
The case when $(P,Q) \in \ \simnet$ is immediate.
Take $(P+R,Q+R) \in \{(P+R,Q+R) \mid P \simnet Q\}$.
Suppose $P+R \xrightarrow{s} T$. 
We want to prove that 
$Q+R \xrightarrow{s'} T'$ with $\reduce{s}=\reduce{s'}$, and $T\mathrel{\mathbf{R}_{sum}} T'$.
There are two cases to be considered, depending on the last SOS rule used, to prove 
$P+R \xrightarrow{s} T$. If the last used rule is
\begin{itemize}
\item \textit{(Rsum)}, then it means that  $R \xrightarrow{s} R'$ for some $R'$ with $T=R'$. 
But then, by using $(Rsum)$, we have $Q+R \xrightarrow{s} T'$, with $T' = R'$ and $T =R' \mathrel{\mathbf{R}_{sum}} R' = T'$ by reflexivity of network bisimilarity $\simnet \subseteq \mathbf{R}_{sum}$.
\item \textit{(Lsum)}, then it means that  $P \xrightarrow{s} P'$ for some $P'$ with $T=P'$.
By assumption, we know that $P \simnet Q$ and therefore there exist $s'$ and $Q'$ such that $Q \xrightarrow{s'}Q'$ with $\reduce{s}=\reduce{s'}$ and $P' \simnet Q'$. 
By applying the rule
$(Lsum)$, we obtain that $Q+R \xrightarrow{s'}Q'$ and we are done.
\end{itemize}

\item[Parallel]
We want to prove that if $P \simnet Q$ then for any $R$ we have $P|R \simnet Q|R$.
We define the relation 
$\mathbf{R}_{par} \defeq \{
(P|R,Q|R) \mid P \simnet Q
\}$ and show that $\mathbf{R}_{par}$ is a network bisimulation.
Suppose $P \simnet Q$ and $P|R \xrightarrow{s} T$. 
We want to prove that 
$Q|R \xrightarrow{s} T'$ with $T\mathrel{\mathbf{R}_{par}} T'$.
There are three cases to be considered, depending on the last SOS rule used to prove $P|R \xrightarrow{s} T$. 
If the last used rule is
\begin{itemize}
\item \textit{(Rpar)}, then it means that $R \xrightarrow{s} R'$ for some $R'$ with $T=P|R'$. 
But then, by using  \textit{(Rpar)}, we have $Q|R \xrightarrow{s} Q|R'$ and $P|R' \mathrel{\mathbf{R}_{par}} Q|R'$, by definition of $\mathbf{R}_{par}$.
\item \textit{(Lpar)}, then it means that $P \xrightarrow{s} P'$ for some $P'$ with $T=P'|R$. 
By assumption, we know that $P \simnet Q$ and therefore there exist $s'$, $Q'$ such that $Q \xrightarrow{s'} Q'$ with $\reduce{s}=\reduce{s'}$ and $P' \simnet Q'$. 
By applying the rule \textit{(Lpar)}, we have that $Q|R \xrightarrow{s'} Q'|R$ and we are done because $P'|R \mathrel{\mathbf{R}_{par}} Q'|R$, 
by definition of $\mathbf{R}_{par}$.
\item \textit{(Com)}, then it means that $P \xrightarrow{s_{1}} P'$, $R \xrightarrow{s_{2}} R'$,  for some $s_{1}, s_{2}, P', R'$ with $s = \merges{s_{1}}{s_{2}}$ and $T = P'|R'$. 
By assumption, we know that $P \simnet Q$ and therefore there exist $s'_{1}$, $Q'$ such that $Q \xrightarrow{s'_{1}} Q'$ with $\reduce{s_{1}}=\reduce{s'_{1}}$ and $P' \simnet Q'$. Now it may be the case that $\merges{s'_{1}}{s_{2}}$ is not defined, but by Lemma~\ref{lem:mergewhite}, we know that $s'_{1}$ and $s_{2}$ can be stretched respectively to $s''_{1}\blackstretcheq s'_{1}$ and $s'_{2}\blackstretcheq s_2$ by inserting or removing enough virtual links to have that $\merges{s''_{1}}{s'_{2}}$ is defined and
$\reduce{\merges{s_{1}}{s_{2}}} = \reduce{\merges{s''_{1}}{s'_{2}}}$.
By the Accordion Lemma~\ref{lem:stretchlts}, $Q \xrightarrow{s''_{1}} Q'$ and $R \xrightarrow{s'_{2}} R'$. 
We conclude by applying rule \textit{(Com)}: $Q|R \xrightarrow{s'} Q'|R'$ with $s' = \merges{s''_{1}}{s'_{2}} \stretcheq s$ and  $P'|R' \mathrel{\mathbf{R}_{par}} Q'|R'$, by definition of $\mathbf{R}_{par}$.
\end{itemize}
\item[Recursion]
Let $E$ and $F$ be two processes that invoke the process identifier $X$.
Let us denote with $E\{P \slash X\}$ the process obtained by replacing in $E$ every occurrence of the identifier $X$ with the process $P$.  
Assume that for any process $P$ we have  $E\{P \slash X\} \simnet F\{P \slash X\}$.
We want to prove that, given  the process definitions
$A\defeq E\{A \slash X\}$, $B\defeq F\{B \slash X\}$, then $A \simnet B$.  
The proof proceeds by showing that:
\begin{enumerate}
\item If $A \defeq Q$ is a process definition, then $A  \simnet Q$.
\item Given the process definitions $A\defeq E\{A \slash X\}$ and $B\defeq F\{B \slash X\}$, for any process $G$ that invokes $X$ we have 
$G\{A \slash X\} \simnet G\{B \slash X\}$.
\end{enumerate}
Then, we have $A \simnet E\{A \slash X\} \simnet E\{B \slash X\}  \simnet F\{B \slash X\} \simnet B$.
The proof of~(1) is immediate by rule \textit{(Ide)}, as $A$ and $Q$ have exactly the same transitions, while the proof of~(2) proceeds in the standard way exploiting induction on derivations, as detailed in the appendix.

\end{description}
\end{proof}

\begin{remark}
As for CCS, it can be proved that several useful axioms over processes hold up to network bisimilarity,
like the commutative monoidal laws for $|$ and $+$, the idempotence of $+$ and the usual laws about restriction.
\end{remark}

\subsection{Semantics Closure with Respect to Substitutions}

One relevant difference w.r.t.\ strong and weak bisimilarity in CCS is that network bisimilarity is also closed with respect to substitutions.

At the level of chains, name substitution is defined as the renaming (see Definition~\ref{def:chainren}).  Not to overload the notation, we denote substitution with $\{-/-\}$.

Given $s=\ell_1 \dots \ell_n$, with $\ell_i = \link{\alpha_1}{\beta_i}$ for $i \in [1,n]$,
we define the substitution of channel $a$ with channel $b$ in a link chain $s$, $s\{\subst{b}{a}\}$ as follows
\[
\begin{array}{c}
\begin{array}{lcl@{\hskip 1cm}lcl}
s\{\subst{b}{a}\}&=& \ell_1\{\subst{b}{a}\} \dots \ell_n\{\subst{b}{a}\} \\
\ell_i\{\subst{b}{a}\} &=& \link{\alpha_i\{\subst{b}{a}\}}{\beta_i\{\subst{b}{a}\}} 
&
\alpha\{\subst{b}{a}\} &=& \left\{\begin{array}{ll} 
b & \mbox{ if } \alpha =a \\
\alpha & \mbox{otherwise}
\end{array}\right.
\end{array}\\
\end{array}
\]

The first observation is that equivalence $\stretcheq$ (but also $\blackstretcheq$) is closed with respect to substitution, as stated by the next lemma.

\begin{restatable}[]{lemma}{lemmaquarantacinque}\label{lem:linksubst} 
For any $a,b,s,s'$
\begin{enumerate}
\item[(i)]
If $s \blackstretcheq s'$ then $s \{\subst{b}{a}\} \blackstretcheq s'\{\subst{b}{a}\}$.
\item[(ii)]
If $s \stretcheq s'$ then $s \{\subst{b}{a}\} \stretcheq s'\{\subst{b}{a}\}$.
\end{enumerate}
\end{restatable}


Next, we prove that transitions are respected by substitutions.
On processes, name substitution differs from renaming and a different notation is used.
Let us denote by $P\{\subst{b}{a}\}$ the simultaneous, capture-avoiding substitution of all the (free) occurrences of $a$ with $b$ in $P$.
Substitution is defined inductively as follows.
\begin{eqnarray*}
\nil \{b/a\} & \defeq & \nil  \\
\ell.P\{b/a\} & \defeq & \ell\{b/a\}.P\{b/a\}\\
(P+Q)\{b/a\} & \defeq & (P\{b/a\}) + (Q\{b/a\})\\
(P|Q)\{b/a\} & \defeq & (P\{b/a\}) | (Q\{b/a\})\\
((\nu c)P)\{b/a\} & \defeq & (\nu d)(P\{d/c\}\{b/a\}) \mbox{ with $d$ fresh} \\
P[\phi]\{b/a\} & \defeq & P\{\phi^{-1}(b) / \phi^{-1}(a)\}[\phi] \\
A(\tilde{c})\{b/a\} & \defeq &  A(\tilde{c}\{b/a\})
\end{eqnarray*}

Substitution enjoys properties similar to that of {renaming}.
In particular, in the proof of Lemma~\ref{lem:redacsubst} we exploit the following property (see Lemma~\ref{lemma:rename}(ii)),
whose proof follows immediately by definition of $\merges{}{}$ and substitution.

\begin{restatable}[]{lemma}{lemmadiciannovebis}\label{lemma:renamebis} 
For any $a,b,s_1,s_2$, if $\merges{s_1}{s_2}$ is defined, then $(\merges{s_1}{s_2})\{\subst{b}{a}\} = \merges{s_1\{\subst{b}{a}\}}{(s_2\{\subst{b}{a}\})}$.
\end{restatable}


\begin{lemma}\label{lem:redacsubst}
For any process $P$ the following holds:
\begin{enumerate}
\item if $P \xrightarrow{s} P'$ then $P\{\subst{b}{a}\}\xrightarrow{s\{\subst{b}{a}\}}P'\{\subst{b}{a}\}$.
\item if $P\{\subst{b}{a}\}\xrightarrow{s}P'$ then there exists $s'$ and $P''$ such that $P\xrightarrow{s'}P''$ with $P' = P''\{\subst{b}{a}\}$ and $s \blackstretcheq s'\{\subst{b}{a}\}$.
\end{enumerate}
\end{lemma} 
\begin{proof}
We prove the two items separately.
\begin{enumerate}
\item The proof is straightforward by rule induction and thus omitted.
\item 
The proof is in two steps.
First we observe that whenever $P\{\subst{b}{a}\}\xrightarrow{s}P'$ then, by the Accordion Lemma~\ref{lem:stretchlts}, $P\{\subst{b}{a}\}\xrightarrow{s''}P'$ with $s \blackstretcheq s''$ where there is no matched occurrences of $b$ in $s''$ (the chain $s''$ is obtained from $s$ by applying the fourth axiom in Definition~\ref{def:black} as many times as needed). Then we prove that if $P\{\subst{b}{a}\}\xrightarrow{s''}P'$ where there is no matched occurrences of $b$ in $s''$, then there exists $s'$ and $P''$ such that $P\xrightarrow{s'}P''$ with $P' = P''\{\subst{b}{a}\}$ and $s'' = s'\{\subst{b}{a}\}$.
The rule proceeds by rule induction as detailed below.
\begin{description}
\item [Rule (Act)]
By hypothesis, $P = \ell.P''$, and we get 
$P\{\subst{b}{a}\} = (\ell.P'')\{\subst{b}{a}\} = \ell\{\subst{b}{a}\}.P''\{\subst{b}{a}\} $
with $P\{\subst{b}{a}\} \xrightarrow{s''}P''\{\subst{b}{a}\}$ and $s'' \blackstretcheq \ell \{\subst{b}{a}\}$.
Moreover, by rule $(Act)$, $P=\ell.P'' \xrightarrow{s'}P'' $ for any $s' \blackstretcheq \ell$.
In particular, we can choose $s'$ to have the same virtual links (and in the same positions) as $s''$, so that 
$s' \{\subst{b}{a}\}= s''$.
By putting $P' = P''\{\subst{b}{a}\}$, we are done.

\item [Rule (Res)]
 By hypothesis, $P = \restrict{c}Q$ and without loss of generality assume that $c \neq a,b$. We get $P\{\subst{b}{a}\} = (\restrict{c}Q)\{\subst{b}{a}\} = \restrict{c}(Q\{\subst{b}{a}\})$. 
 Hence, there exist $Q'$ and $s''_1$  such that 
 \[P\{\subst{b}{a}\}=\restrict{c}(Q\{\subst{b}{a}\}) \xrightarrow{\restrict{c}s''_1 }\restrict{c}Q'=P'
 \]
 with $s''=\restrict{c}s''_1$ and $Q\{\subst{b}{a}\}\xrightarrow{s''_1}Q'$.
 Since there is no matched occurrence of $b$ in $s''$, there is none in $s''_1$.
 By inductive hypothesis, there exist $Q''$ and $s'_1$ such that
 $Q \xrightarrow {s'_1} Q''$, with $s'_1\{\subst{b}{a}\} = s''_1$, and $Q' = Q''\{\subst{b}{a}\} $.
Since $\restrict{c}s''_1 = \restrict{c}(s'_1\{\subst{b}{a}\})$ is defined then also $\restrict{c}s'_1$ is defined, because $c\neq a,b$.
Thus, we can apply rule $(Res)$ to get
 $\restrict{c}Q \xrightarrow{\restrict{c}s'_1}\restrict{c}Q'' $, and we 
take $P''=\restrict{c}Q''$ and $s'=\restrict{c}s'_1$.
Then we get  
$P' = \restrict{c}Q' = \restrict{c}(Q''\{\subst{b}{a}\})= P''\{\subst{b}{a}\}$ and 
$s''= \restrict{c}s''_1= \restrict{c}(s'_1\{\subst{b}{a}\})= (\restrict{c}s'_1)\{\subst{b}{a}\}=s'\{\subst{b}{a}\}$.

\item [Rule (Com)]
 By hypothesis, $P= R|Q$ and, by rule (Com), we have
 \[P\{\subst{b}{a}\} =(R | Q)\{\subst{b}{a}\} = (R\{\subst{b}{a}\} | Q\{\subst{b}{a}\})\xrightarrow{\merges{s''_1}{s''_2}} (R'|Q') = P'
 \]
 with $s''=\merges{s''_1}{s''_2}$, $R\{\subst{b}{a}\} \xrightarrow{s''_1}R'$ and $Q\{\subst{b}{a}\} \xrightarrow{s''_2}Q'$.
Since all there is no matched occurrence of $b$ in $s''$, there is none in both $s''_1$ and $s''_2$. 
By inductive hypothesis, there exist $R''$, $s'_1$, $Q''$, $s'_2$ such that
 $R \xrightarrow{s'_1}R''$ and  $Q \xrightarrow{s'_2}Q''$, with 
 $s''_1 =  s'_1\{\subst{b}{a}\}$,  $R' = R''\{\subst{b}{a}\}$, $s''_2 =  s'_2\{\subst{b}{a}\}$,  $Q' = Q''\{\subst{b}{a}\}$. 
Now we observe that $\merges{s'_1}{s'_2}$ is defined. In fact the only reason for which $\merges{s'_1}{s'_2}$ is undefined when $\merges{s'_1\{\subst{b}{a}\}}{(s'_2\{\subst{b}{a}\})}$ is defined would be that an action $a$ should be paired with an action $b$ before the substitution takes place, but this is ruled out by the assumption that there is no matched occurrence of $b$ in $s''$.
 By rule $(Com)$, we have $R|Q \xrightarrow{\merges{s'_1}{s'_2}} R''|Q''$.
 Now, we take $P''=R''|Q''$, $s'=\merges{s'_1}{s'_2}$ and we get $P' = R'|Q' = R''\{\subst{b}{a}\}|Q''\{\subst{b}{a}\}=(R''|Q'')\{\subst{b}{a}\}=P''\{\subst{b}{a}\}$,
  $s''=\merges{s''_1}{s''_2} = \merges{s'_1\{\subst{b}{a}\}}{(s'_2\{\subst{b}{a}\})} = (\merges{s'_1}{s'_2})\{\subst{b}{a}\}= s'\{\subst{b}{a}\}$ and we are done. 

\item [Rule (Ren)]
By hypothesis, $P = Q[\phi]$. 
Let $a' = \phi^{-1}(a)$ and $b' = \phi^{-1}(b)$.
We get that
$P\{\subst{b}{a}\} = Q[\phi]\{\subst{b}{a}\}
=Q\{\subst{b'}{a'}\}[\phi]$
 and, by rule $(Ren)$ there exist $Q'$ and $s''_1$ such that
\[P\{\subst{b}{a}\} = Q[\phi]\{\subst{b}{a}\} = Q\{\subst{b'}{a'}\}[\phi]
 \xrightarrow{s''_1[\phi]}  Q'[\phi] = P'
\]
with $s'' =s''_1[\phi]$ and
 $Q\{\subst{b'}{a'}\} \xrightarrow {s_1} Q'$. 
 Since there is no matched occurrence of $b$ in $s''$ it must be the case that there is no matched occurrence of $b'$ in $s''_1$.
 By inductive hypothesis, there exists $Q''$ and $s'_1$ 
 s.t.~$Q \xrightarrow {s'_1} Q''$, with $s''_1 = s'_1\{\subst{b'}{a'}\}$, and $Q' = Q''\{\subst{b'}{a'}\}$.
By rule $(Ren)$, $Q[\phi] \xrightarrow{s'_1[\phi]} Q''[\phi] $, and we 
take $P''=Q''[\phi]$ and $s'=s'_1[\phi]$.
Then we get  
$P' =Q'[\phi] = Q''\{\subst{b'}{a'}\}[\phi]= Q''[\phi]\{\subst{b}{a}\} =
P''\{\subst{b}{a}\}$ and
$s'' = s''_1[\phi] = s'_1\{\subst{b'}{a'}\}[\phi] = s'_1[\phi]\{\subst{b}{a}\} = s'\{\subst{b}{a}\}$.
\end{description}
For the remaining rules the proofs are simpler and thus omitted.
\end{enumerate}
\end{proof}

\begin{proposition}\label{prop:subs}
For any processes $P,Q$, if $P \simnet Q$ then for any channel $a,b$ we have $P\{\subst{b}{a}\} \simnet Q\{\subst{b}{a}\}$.
\end{proposition}
\begin{proof}
We define the relation $\mathbf{R}_{sub} = \{ (P\{\subst{b}{a}\},Q\{\subst{b}{a}\})\;|\; P \simnet Q \}$ and prove that it is a network bisimulation.\\
We want to show that if $P\{\subst{b}{a}\} \xrightarrow{s} P'$ then there exist $Q'$ and $s'$ such that $Q\{\subst{b}{a}\} \xrightarrow{s'}Q'$, with $s \stretcheq s'$ and 
$P' \ \mathbf{R}_{sub}\ Q'$.\\
By Lemma~\ref{lem:redacsubst} (point 2), there exist $P''$ and $s''$ such that 
$P \xrightarrow{s''}P''$ with $P' = P''\{\subst{b}{a}\}$ and $s  \blackstretcheq s''\{\subst{b}{a}\}$.\\
By hypothesis, $P \simnet Q$, then $\exists$ $s'''$, $Q''$ such that $Q\xrightarrow{s'''} Q''$, with $s'' \stretcheq s'''$ and $P'' \simnet Q'''$.\\
By Lemma~\ref{lem:redacsubst} (point 1),  $Q\{\subst{b}{a}\}\xrightarrow{s'''\{\subst{b}{a}\}} Q''\{\subst{b}{a}\}$, and we take $s' = s'''\{\subst{b}{a}\}$ and $Q' = Q''\{\subst{b}{a}\}$. Now we are done, since  $P' = P''\{\subst{b}{a}\} \ \mathbf{R}_{sub} \ Q''\{\subst{b}{a}\}=Q'$ and, 
by Lemma~\ref{lem:linksubst}(ii),  $ s \blackstretcheq s''\{\subst{b}{a}\} \stretcheq s'''\{\subst{b}{a}\}=s'$.
\end{proof}

\begin{example}
It is illustrative to revisit the classical CCS counterexample that shows that strong bisimilarity is not a congruence w.r.t. \ substitution, already mentioned in Section~\ref{sec:CCS}. The translations of the two CCS processes $a.\nil \ |\ \overline{b}.\nil \sim a.\overline{b}.\nil + \overline{b}.a.\nil$ in \CNA\ are respectively
$\link{\silent}{a}.\nil\ |\ \link{b}{\silent}.\nil$ and $\link{\silent}{a}.\link{b}{\silent}.\nil + \link{b}{\silent}.\link{\silent}{a}.\nil$. But now the transition 
$\link{\silent}{a}.\nil\ |\ \link{b}{\silent}.\nil 
\xrightarrow{\startchain{\silent}\chainedlink{a}{\noact}\chainedlink{\noact}{b}\chainend{\silent}}
\nil\ |\ \nil$ cannot be simulated by $\link{\silent}{a}.\link{b}{\silent}.\nil + \link{b}{\silent}.\link{\silent}{a}.\nil$ and
the two processes are not network bisimilar.
The reason the transition cannot be simulated is that in one case you have concurrency and in the other sequentiality, while in our semantics rule Com applies in the concurrent case.
\end{example}

\subsection{Composite and Dynamic Routing in \CNA}

By using \CNA\ as a modelling framework, we can now revisit the example of composite routing  and prove some interesting properties. We start by introducing the simplest possible algebra for building complex routing infrastructures starting from basic building blocks. As done in Section~\ref{sec:CCS}, we can imagine a routing infrastructure as a box with a left and right interface and with some connections from (some of) the left channels to (some of) the right channels.

\begin{definition}[Basic infrastructure]
Let $\widetilde{a}= a_1,...,a_n$ and $\widetilde{b} = b_1,...,b_m$ be two lists of channels. A \emph{basic routing infrastructure} $R$ from $\widetilde{a}$ to $\widetilde{b}$, written $R(\widetilde{a},\widetilde{b})$ is a \CNA\ process of the form
\[
R(\widetilde{a},\widetilde{b}) \defeq \ell_1.R(\widetilde{a},\widetilde{b}) + ... + \ell_k.R(\widetilde{a},\widetilde{b})
\]
with $\ell_h = \link{a_{i_h}}{b_{j_h}}$  where $i_h \in [1,n]$ and $j_h \in [1,m]$ for any $h\in [1,k]$.
\end{definition}

\begin{definition}[Composite infrastructure]
A \emph{composite infrastructure} $R(\widetilde{a},\widetilde{b})$ is either a basic infrastructure or the composition 
\[
R(\widetilde{a},\widetilde{b}) = \restrict{\widetilde{c}}(Q(\widetilde{a},\widetilde{c})\ |\ S(\widetilde{c},\widetilde{b}))
\]
of two (composite) infrastructures $Q(\widetilde{a},\widetilde{c})$ and $S(\widetilde{c},\widetilde{b})$.
\end{definition}

To each (composite) infrastructure $R(\widetilde{a},\widetilde{b})$ we can associate a graph $\mathcal{G}(R(\widetilde{a},\widetilde{b}))$, whose nodes are the channels appearing in the definition\footnote{Without loss of generality, we can exploit alpha-conversion to assume that all restricted channels are named in a different way.} of the process $R(\widetilde{a},\widetilde{b})$, and whose arcs are induced by the links appearing in the process,  i.e.\  there is an arc $x\rightarrow y$ if the link $\link{x}{y}$ appears as a prefix in the definition of $R(\widetilde{a},\widetilde{b})$. We then have the following characterisation of the transitions admitted by $R(\widetilde{a},\widetilde{b})$.

We recall that with $||s||$ we denote the \emph{size} of $s$, i.e.\ the number of solid links in the link chain $s$. Note that size is preserved by the equivalence $\blackstretcheq$, but not by $\stretcheq$.

\begin{restatable}[]{lemma}{lemmacinquantuno}\label{lem:esempio} 
Let $R(\widetilde{a},\widetilde{b})$ be a composite infrastructure. 
\begin{enumerate}
\item If 
$R(\widetilde{a},\widetilde{b}) \xrightarrow{s} R'$ then $R' = R(\widetilde{a},\widetilde{b})$ and there exist two nodes $a_i\in \widetilde{a}$ and $b_j\in \widetilde{b}$ of the graph $\mathcal{G}(R(\widetilde{a},\widetilde{b}))$ such that $s \stretcheq \link{a_i}{b_j}$ and there is a path from $a_i$ to $b_j$ whose length is $||s||$ in the graph $\mathcal{G}(R(\widetilde{a},\widetilde{b}))$.
\item If there is a path of length $n$ from $a_i$ to $b_j$ in the graph $\mathcal{G}(R(\widetilde{a},\widetilde{b}))$ then $R(\widetilde{a},\widetilde{b}) \xrightarrow{s}R(\widetilde{a},\widetilde{b})$  with $s \stretcheq \link{a_i}{b_j}$ and $||s|| = n$.
\end{enumerate}
\end{restatable}

From the previous lemma, it follows that any composite infrastructure $R(\widetilde{a},\widetilde{b})$ is network bisimilar to a basic infrastructure that has one link for each possible path in $\mathcal{G}(R(\widetilde{a},\widetilde{b}))$ from one of the $a_i$s to one of the $b_j$s.

\begin{definition}
Let $R(\widetilde{a},\widetilde{b})$ be a composite infrastructure and $\mathcal{G} = \mathcal{G}(R(\widetilde{a},\widetilde{b}))$ be its corresponding graph. We denote with $P_{\mathcal{G}}(\widetilde{a},\widetilde{b})$ the basic infrastructure that offers one link for any path in $\mathcal{G}$,  i.e.\ 
\[
P_{\mathcal{G}}(\widetilde{a},\widetilde{b}) \defeq \sum_{a_i \rightarrow^* b_j \in \mathcal{G}} \link{a_i}{b_j}.P_{\mathcal{G}}(\widetilde{a},\widetilde{b})  .
\]
where $a \rightarrow^* b$ denotes the presence of a path from $a$ to $b$. 
\end{definition}

\begin{corollary}
Any composite infrastructure $R(\widetilde{a},\widetilde{b})$ is network bisimilar to the basic infrastructure $P_{\mathcal{G}(R(\widetilde{a},\widetilde{b}))}(\widetilde{a},\widetilde{b})$.
\end{corollary}

\begin{example}
Let us define the following basic infrastructures
\begin{eqnarray*}
R'(\widetilde{\mathit{req}},\widetilde{s}) & \defeq &
\link{\mathit{req}_1}{s_1}.R'(\widetilde{\mathit{req}},\widetilde{s}) +
\link{\mathit{req}_1}{s_2}.R'(\widetilde{\mathit{req}},\widetilde{s}) +
\link{\mathit{req}_2}{s_2}.R'(\widetilde{\mathit{req}},\widetilde{s}) \\
R''(\widetilde{s},\widetilde{s'}) & \defeq &
\link{s_1}{s'_1}.R''(\widetilde{s},\widetilde{s'}) +
\link{s_2}{s'_2}.R''(\widetilde{s},\widetilde{s'}) \\
R'''(\widetilde{s'},\widetilde{\mathit{srv}}) & \defeq &
\link{s'_2}{\mathit{srv}_2}.R'''(\widetilde{s'},\widetilde{\mathit{srv}})
\end{eqnarray*}
and combine them to form the composite infrastructures
\begin{eqnarray*}
Q(\widetilde{\mathit{req}},\widetilde{s'}) & \defeq &
\restrict{\widetilde{s}}(R'(\widetilde{\mathit{req}},\widetilde{s}) | R''(\widetilde{s},\widetilde{s'})) \\
R = R(\widetilde{\mathit{req}},\widetilde{\mathit{srv}}) & \defeq &
\restrict{s'}(Q(\widetilde{\mathit{req}},\widetilde{s'})  | R'''(\widetilde{s'},\widetilde{\mathit{srv}}))
\end{eqnarray*}
Assuming all of the tuples $\widetilde{\mathit{req}}$, $\widetilde{s}$, $\widetilde{s'}$ and $\widetilde{\mathit{srv}}$ have length $2$, the graph $\mathcal{G}(R)$ is depicted below
\[
\xymatrix{
{_{\mathit{req}_1}\bullet} \ar[r] \ar[rd] &
{_{s_1}\bullet} \ar[r] &
{\bullet_{s'_1}} &
{\bullet_{\mathit{srv}_1}} 
\\
{_{\mathit{req}_2}\bullet} \ar[r] &
{_{s_2}\bullet} \ar[r] &
{\bullet_{s'_2}} \ar[r] &
{\bullet_{\mathit{srv}_2}} 
}
\]
Then, it is immediately evident that the admissible transitions for $R$ are of the form
\[
R \xrightarrow{
\startchain{\mathit{req}_1}
\chainedlink{\silent}{\silent}
\chainedlink{\silent}{\silent}
\chainend{\mathit{srv}_2}
}
R
\qquad
R \xrightarrow{
\startchain{\mathit{req}_2}
\chainedlink{\silent}{\silent}
\chainedlink{\silent}{\silent}
\chainend{\mathit{srv}_2}
}
R
\]
where, of course, many additional virtual links can be appended to the extremes of the labels (remember the Accordion Lemma~\ref{lem:stretchlts}). 
Consequently, $R$ is network bisimilar to the basic infrastructure
$S(\widetilde{\mathit{req}},\widetilde{\mathit{srv}}) \defeq \link{\mathit{req}_1}{\mathit{srv}_2}.S(\widetilde{\mathit{req}},\widetilde{\mathit{srv}}) + \link{\mathit{req}_2}{\mathit{srv}_2}.S(\widetilde{\mathit{req}},\widetilde{\mathit{srv}})$.
\end{example}

Finally, we show that \CNA\ is particularly convenient to model programmable infrastructures, where links can be dynamically added and removed.

Given $\widetilde{a}= a_1,...,a_n$ and $b = b_1,...,b_m$ let us consider the processes
\begin{eqnarray*}
\widehat{R}_{i,j} & \defeq & \link{add_{i,j}}{\silent}.R_{i,j} \\
R_{i,j} & \defeq & \link{a_i}{b_j}.R_{i,j}\ +\ \link{rem_{i,j}}{\silent}.\widehat{R}_{i,j}\ +\ \link{add_{i,j}}{\silent}.(R_{i,j} | R_{i,j})
\end{eqnarray*}
and their parallel composition
\[
R = \prod_{i=1}^n \prod_{j=1}^{m} \widehat{R}_{i,j}
\]
where we use the shorthand $\prod_{i=1}^n P_i$ for the parallel composition $P_1\ |\ \cdots\ |\ P_n$.

The idea is that an interaction involving the link $\link{add_{i,j}}{\silent}$ allows us to add one link from $a_i$ to $b_j$ and that an interaction involving the link $\link{rem_{i,j}}{\silent}$ allows us to remove one such link.
Several links between $a_i$ and $b_j$ can be available at the same time, but no such link can be removed if it is not present.
Initially, the process $R$ makes no link available.

We believe that modelling infrastructures at this level of abstraction drastically improves the situation w.r.t~other process algebras based on dyadic interaction, such as CCS. In fact imagine the situation where a composite programmable infrastructure is modelled in CCS: it can happen that a transfer of information is started along a viable path, but during the chain of interactions one or more of the hops are removed. As a consequence it is then impossible to deliver the request as well as acknowledge the failure. The \CNA\ middleware guarantees that none of these troublesome scenarios can arise in the model.

\subsection{Alternative Definitions}\label{subsec:altdef}

The theory of \CNA\ is quite strong and stable and it can be extended in several directions without much efforts.
Here we briefly discuss only three noteworthy possible variations of the presented framework.

\paragraph{Chain as prefixes}
In the first variation, we could extend the syntax of \CNA\ to allow essential chains instead of solid links as prefixes,  i.e.\  the grammar production $P ::= \ell.P$ can be replaced by $P ::= s.P$ with $s$ essential. This change can increase the usability of the process algebra in modelling different scenarios. For example, we can write a process such as $\startchain{\silent}\chainedlink{a}{\noact}\chainedlink{\noact}{b}\chainedlink{b}{\noact}\chainedlink{\noact}{c}\chainend{\silent}.P$ that requires an interaction from $a$ to $c$ via $b$. All of the results presented in the paper would carry over such an extension. Remarkably, network bisimilarity would still be a congruence.

\paragraph{Bisimilarity $\simbnet$}
In the second variation, in the bisimilarity game, we could decide to take into account the number of traversed (solid) links so to get a finer equivalence. This amounts to changing the definition of bisimulation by requiring that the matching label $s'$ is related to $s$ by $\blackstretcheq$ instead of $\stretcheq$. We can denote the corresponding bisimilarity as $\simbnet$. Since $\blackstretcheq\ \subseteq\ \stretcheq$, it follows that $\simbnet$ is finer than $\simnet$,  i.e.\ , it distinguishes more processes. However, as in the previous case, all the results presented in the paper would carry over this change.

\paragraph{Ordinary bisimilarity $\sim = \simbnet$}
In the third and last variation, we could take ordinary strong bisimilarity $\sim$, by requiring exact matching of labels. Then, because of the Accordion Lemma~\ref{lem:stretchlts}, the resulting equivalence would coincide with the equivalence $\simbnet$ from the second point.



\section{Concluding Remarks and Related Works}
\label{sec:conclusion}
In this paper we have presented \CNA\ as a generalisation of traditional dyadic process calculi able to deal with open multiparty interactions.
These more complex forms of interactions can be represented in \CNA\, without complicating
the underlying synchronisation algebra, still quite simple and with rules similar to the ones of CCS.
We have provided the calculus with an abstract semantics, called network bisimilarity that, as the strong bisimilarity of CCS, is a congruence
w.r.t.\ all the operators of \CNA. In addition, network bisimilarity is also a congruence w.r.t.\ substitutions.
Furthermore, the theory of \CNA\ is quite stable under several variations, such as allowing (essential) link chains as prefixes or changing the notion of network bisimulation to get finer equivalences.

Formally capturing new patterns of communication seems crucial to understand today's Internet infrastructures
and their intrinsic dynamic nature.
From this point of view, we have shown that \CNA \ is particularly convenient for modelling programmable
infrastructures, where links can be added and removed in a dynamic way.

Concerning the taxonomy for multiparty  languages
proposed in~\cite{JS96}, we can say that \CNA\
is  {\em variable}, i.e.\ the number of participants is not fixed a priori, and
{\em asynchronous}, i.e.\ not all the processes in the systems are required 
to make a move at each step. 
In contrast,  \CNA\
adopts a multi-channel mechanism that does not
work as a {\em gate} forcing all the involved processes
to take part in the interaction.
We can say that our {\em  interaction command}, i.e. the command used to
establish a multiparty interaction, only allows a multiparty interaction to happen.
 Thus, following this taxonomy we can classify \CNA\
neither as {\em conjunctive} nor as  {\em disjunctive} calculus.

As stated in Section~\ref{subsec:altdef}, we intend to take \CNA\ as a starting point for investigating more general forms of interaction and more advanced forms of equivalence. Several interesting directions are possible.

Some alternatives to network bisimilarity have been discussed in Section~\ref{subsec:altdef}. 
Weak variants of them can be readily defined by considering solid link chains as internal (silent) actions
(as they represent completed interactions).
As usual, the corresponding equivalences will not be congruences w.r.t.\ choice.
However, we think that the multi-party interaction available in \CNA\ offers already a more abstract mechanism than
dyadic communication, so that weak equivalences are not needed for many applications.

Another possibility, frequently used in process calculi literature, is to define the operational semantics in terms of reductions and then derive (context-closed) observational equivalences on the basis of some well-chosen observables. 
While the obvious choice for the observables would be link prefixes, it is difficult to set up the same methodology for \CNA\ because open multi-party interactions can involve an unbounded number of participants and are difficult to model as reduction rules.
Nevertheless, this would be an interesting research direction to explore in the future.

The name handling variant of \CNA, called \emph{link-calculus}, has been already considered in~\cite{BodeiBB12}
and exploited in~\cite{BodeiBBC14} to model biological interactions.

%
Due to space limitation, we decided to focus here on presenting the communication layer in full details and devote a companion paper to the name handling extension, which is currently under scrutiny.
In particular, on the more applicative side, we think that the generalisation of link prefixes to {\em link-chain} prefixes can be very useful to encode some simple patterns of interaction directly in the action prefixes, thus enhancing the modelling power.

We plan to extend the theory to take into account some weights associated with each link, along the lines of~\cite{DBLP:journals/corr/abs-1011-6308}. For example, if weights are seen as costs, then processes can be compared on the basis of the overall cost of an interaction they offer, and the abstract equivalence can be refined to a preorder to reflect the fact that when two processes offer the same interactions, one is cheaper than the other.
If costs are replaced with some logical information, e.g.\ representing the knowledge associated with the link, then an interaction can be paired with deduction and thus compared with others on the basis of the amount of information it provides.
Other quantitative extensions could exploit probabilities and stochastic rates.

Another direction for future work is concerned with the cross-fertilisation between computational sciences and biology.
In~\cite{BodeiBBC14} we have shown that membrane interactions are intrinsically multi-party, by providing a faithful encoding of Brane calculi~\cite{C05} in terms of \emph{link-calculus}. Brane calculi are compartment-based calculi, introduced to model the behaviour of nested membranes in complex biological systems. We plan to include causality in the picture, so to study dependencies among interactions and track down sources of unwanted behaviours and consequences of biological reactions. 
Causality enriched models have already been used to study metabolic networks~\cite{BodeiBC08,BarbutiGM17,BarbutiGLM18}, e.g.\ for detecting incorrect behaviour that may depend on a particular
ordering of certain events, sometimes difficult to predict.
The idea is to define a causal semantics for \CNA \ and exploit static analysis techniques for approximating the causal relationships among the interactions performed by a complex system, along the lines
of the abstract causal semantics proposed in~\cite{BGL13,BGL15,BBGHL15,BBGLBH17} for the Brane calculus and of the context dependent analysis presented in~\cite{newcfaBioAmb} for BioAmbients~\cite{Bio_Amb}, another calculus for describing biological systems.

A further extension of our approach consists in the possibility of expressing 
non linear communication patterns in the prefixes, as allowing links of arity greater than 2 and combine them in trees, matrices or graphs.


\paragraph{Related Work}
Among the recently presented network-aware extensions of classical calculi such as~\cite{FrancalanzaH08} (to handle explicit distribution, remote operations and process mobility), 
and~\cite{NicolaGP07} (to deal with permanent nodes crashing and links breaking), the
closest proposal to ours is in~\cite{NCPibis}, an extension of  
$\pi$-calculus, where links are named and are distinct from usual input/output actions, and there is one sender and one receiver (the output includes the final receiver name). 
In the name-passing variant of \CNA~\cite{BodeiBB12}, links can carry message tuples, and each participant can play both the sender and the receiver r\^{o}le.
This extended semantics recalls the concurrent semantics in~\cite{NCPibis}, where concurrent transmissions can be observed in the form of a multi-set of routing paths. In our case the collected links are organised in a link chain.

In~\cite{BL08}, the authors present a general framework to extend synchronisation algebras~\cite{W84} with name mobility that 
could be easily adapted to many other high-level kinds of synchronisation, like the one we need, but with a more complex machinery.
%
More sophisticated forms of synchronisations, with a fixed number of processes, are introduced in $\pi$-calculus 
in~\cite{N98} (joint input) and in~\cite{CM03} (polyadic synchronisation).
%
%
The focus of~\cite{LV10} is instead on the expressiveness of an asynchronous CCS equipped with joint inputs allowing the interactions of $n$ processes, proving that there is no truly distributed implementation of operators synchronising more than three processes. 
As in the Join-calculus~\cite{FG96}, and differently from our approach, participants can act either as senders or as receivers.
%

In~\cite{GorrieriV11}, a conservative extension of CCS, called A$^2$CCS, is studied that is 
able to model multi-party synchronisation. 
The mechanism 
is realised as an atomic sequence of dyadic synchronisations of arbitrary lengths, but imposes some constraints 
that make the parallel operator non associative and therefore more difficult to use as a model.

Finally, in~\cite{BodeiDP01}, a distributed version of the $\pi$-calculus 
for handling names, considered as localised to their owners, in concurrent and distributed systems made of mobile processes.
Each process is indeed equipped with a local name environment.
When a name is exported, it is equipped with the information needed to point back to its local environment, thus keeping track of the origin of mobile agents in multi-hop travel on the network.
Communications are not open, but are instead controlled by a distributed name manager that keeps distinct the names generated by different environments.

As a last remark, it is worth noting that the operational semantics of \CNA \ allows a link prefix to participate 
in infinitely many transitions that account for the presence of the link within chains of any length.
Thus, a direct implementation of the \CNA \ semantics that can be used for verification is not immediate.
A possible solution to overcome this problem is the definition of a symbolic semantics. The one in~\cite{BrodoO17} collapses in a single transition all the transitions labelled with link chains composed with the same set of solid links. Its implementation can be found in~\cite{tool}, where an online tool is available for the simulation for CNA-encoded examples.

\paragraph{Acknowledgements}
We would like to thank the reviewers for their careful comments and suggestions that helped us to improve both the presentation and the technical contents of the paper.


\bibliography{biblio_CNA}
\bibliographystyle{plain}

\appendix


\section{Proofs of Technical Results}\label{app}
In this appendix we restate the lemmata presented earlier in the paper and gives the proofs of their correctness.

\subsection{Proofs of Section~\ref{sec:chains}}
We recall Lemma~\ref{lemma:solid}
(the original appears on p.\,\pageref{lemma:solid}).
\lemmaundici*

\begin{proof} 
We prove the three items separately.
\begin{enumerate}
\item[(i)]
Trivial, by the fact that the underlying operation on actions $\merges{\alpha}{\beta}$ is commutative and associative.
\item[(ii)]
The thesis follows by applying the definition of merge.
Let $\ell = \link{\alpha}{\beta}$ and $\ell' = \link{\alpha'}{\beta'}$ then 
$\merges{\ell}{\ell'} = \link{(\merges{\alpha}{\alpha'})}{(\merges{\beta}{\beta'})}$,
where $(\merges{\alpha}{\alpha'}) = \ \noact$ iff $\alpha = \alpha' = \ \noact$. 
Similarly, $(\merges{\beta}{\beta'}) = \ \noact$ iff $\beta = \beta' = \ \noact$.
\item[(iii)]
Since $s$ is solid, its length $n=|s|$ is greater than zero.
If $|s'| \neq n$ then  $\merges{s}{s'}=\bot$.
Otherwise, let $s=\ell_{1}...\ell_{n}$ and $s'=\ell'_{1}...\ell'_{n}$.
Since link chains cannot be made of virtual links only, there is at least a position $i$ in $s'$ such that $\ell'_{i}$ is solid.
Then, $\merges{\ell_{i}}{\ell'_{i}} = \bot$ because, since $s$ is solid, also $\ell_{i}$ is solid.
As a consequence, $\merges{s}{s'}=\bot$.
\end{enumerate}
\end{proof}

We recall Lemma~\ref{lem:mergestretch}
(the original appears on p.\,\pageref{lem:mergestretch}).
\lemmadodici*

\begin{proof}
We prove that the axioms in Definition~\ref{def:black}, when applied in either direction, satisfy the property. Then, by transitivity, the thesis holds for all the elements in each equivalent class of $\blackstretcheq$.
The proof proceeds by cases on the axioms of  $\blackstretcheq$.

\begin{description}
\item[{\bf case $[s_0 \blackstretcheq s_0 \link{\noact}{\noact}]$}]
Let $\merges{s'}{s''}=s_0$ and $s = s_0 \link{\noact}{\noact}$. 
Now we have to find $s_1\blackstretcheq s'$ and $s_2\blackstretcheq s''$ such that $\merges{s_1}{s_2} = s$. 
To this aim, we
set\footnote{Note that $s_1$ and $s_2$ are link chains since otherwise $(\merges{s'}{s''})\link{\noact}{\noact}= s$ would not be defined.}  $s_1 = s'\link{\noact}{\noact}$ and $s_2 = s''\link{\noact}{\noact}$. 
By definition of the merge operator, $\merges{}{}$,   
we get that 
$\merges{s_1}{s_2} 
= \merges{s'\link{\noact}{\noact}}{s''\link{\noact}{\noact}} 
= (\merges{s'}{s''}) \link{\noact}{\noact} 
= s_0 \link{\noact}{\noact}
= s$.

\item[{\bf case $[s_0 \link{\noact}{\noact} \blackstretcheq s_0]$}]
By hypothesis, $\merges{s'}{s''}=s_0\link{\noact}{\noact}$ and $s = s_0$. By definition of  the merge operator, $\merges{}{}$, we get $s'=s_1\link{\noact}{\noact}$ and $s'' = s_2 \link{\noact}{\noact}$, for suitable $s_1$ and $s_2$.
Then, we have $s=\merges{s_1}{s_2}$.

\item[{\bf case $[s_0 \startchain{\noact}\chainedlink{\noact}{\noact}\chainend{\noact}s'_0 \blackstretcheq  s_0 \link{\noact}{\noact} s'_0]$}] 
Let 
$\merges{s'}{s''} = s_0 \startchain{\noact}\chainedlink{\noact}{\noact}\chainend{\noact}s'_0$ and 
$s=s_0 \link{\noact}{\noact} s'_0$. 
Therefore it must be the case that 
$s' = s'_1 \startchain{\noact}\chainedlink{\noact}{\noact}\chainend{\noact}s''_1$ and 
$s'' = s'_2 \startchain{\noact}\chainedlink{\noact}{\noact}\chainend{\noact}s''_2$ 
with $\merges{s'_1}{s'_2} = s_0$ and $\merges{s''_1}{s''_2} = s'_0$. Then we let 
$s_1 = s'_1 \link{\noact}{\noact}s''_1$ and 
$s_2 = s'_2 \link{\noact}{\noact}s''_2$ and we get
$\merges{s_1}{s_2} = 
(\merges{s'_1}{s'_2})\link{\noact}{\noact}(\merges{s''_1}{s''_2}) =
s_0\link{\noact}{\noact}s'_0 = s$.

\item[{\bf case $[s_0 \link{\noact}{\noact} s'_0 \blackstretcheq s_0 \startchain{\noact}\chainedlink{\noact}{\noact}\chainend{\noact}s'_0]$}] 
Let 
$\merges{s'}{s''} = s_0 \link{\noact}{\noact} s'_0$ and 
$s = s_0 \startchain{\noact}\chainedlink{\noact}{\noact}\chainend{\noact}s'_0$. 
Therefore it must be the case that 
$s' = s'_1 \link{\noact}{\noact} s''_1$ and 
$s'' = s'_2 \link{\noact}{\noact} s''_2$ 
with $\merges{s'_1}{s'_2} = s_0$ and $\merges{s''_1}{s''_2} = s'_0$. Then we let 
$s_1 = s'_1 \startchain{\noact}\chainedlink{\noact}{\noact}\chainend{\noact}s''_1$ and 
$s_2 = s'_2 \startchain{\noact}\chainedlink{\noact}{\noact}\chainend{\noact}s''_2$ and we get
$\merges{s_1}{s_2} = 
(\merges{s'_1}{s'_2})\startchain{\noact}\chainedlink{\noact}{\noact}\chainend{\noact}(\merges{s''_1}{s''_2}) =
s_0\startchain{\noact}\chainedlink{\noact}{\noact}\chainend{\noact}s'_0 = s$.
\item[{\bf case $[s_0 \startchain{\alpha}\chainedlink{a}{\noact}\chainedlink{\noact}{a}\chainend{\beta}s'_0 \blackstretcheq  s_0 \startchain{\alpha}\chainedlink{a}{a}\chainend{\beta} s'_0]$}]
Let $\merges{s'}{s''}=s_0\startchain{\alpha}\chainedlink{a}{\noact}\chainedlink{\noact}{a}\chainend{\beta}s'_0 $ and $s = s_0 \startchain{\alpha}\chainedlink{a}{a}\chainend{\beta} s'_0$.
Then, there are four possible cases:
\begin{description}
\item $s' = s'_1\startchain{\noact}\chainedlink{\noact}{\noact}\chainedlink{\noact}{\noact}\chainend{\noact}s''_1$ and $s'' =s'_2\startchain{\alpha}\chainedlink{a}{\noact}\chainedlink{\noact}{a}\chainend{\beta}s''_2$
\item $s' = s'_1\startchain{\alpha}\chainedlink{a}{\noact}\chainedlink{\noact}{\noact}\chainend{\noact}s''_1$ and $s'' = s'_2\startchain{\noact}\chainedlink{\noact}{\noact}\chainedlink{\noact}{a}\chainend{\beta}s''_2$
\item $s' = s'_1\startchain{\noact}\chainedlink{\noact}{\noact}\chainedlink{\noact}{a}\chainend{\beta}s''_1$ and $s'' = s'_2\startchain{\alpha}\chainedlink{a}{\noact}\chainedlink{\noact}{\noact}\chainend{\noact}s''_2$
\item  $s' = s'_1\startchain{\alpha}\chainedlink{a}{\noact}\chainedlink{\noact}{a}\chainend{\beta}s''_1$ and $s'' = s'_2\startchain{\noact}\chainedlink{\noact}{\noact}\chainedlink{\noact}{\noact}\chainend{\noact}s''_2$
\end{description}
with $\merges{s'_1}{s'_2}=s_0$ and $\merges{s''_1}{s''_2}=s'_0$.\\
We only show the first case, as the other ones are similar.\\
Now we let $s_1 = s'_1 \startchain{\noact}\chainedlink{\noact}{\noact}\chainend{\noact}s''_1 \blackstretcheq s'$ and 
$s_2 = s'_2 \startchain{\alpha}\chainedlink{a}{a}\chainend{\beta}s'_2 = s$ and we are done 
since
$\merges{s_1}{s_2}=(\merges{s_1}{s'_2}) \startchain{\alpha}\chainedlink{a}{a}\chainend{\beta}(\merges{s''_1}{s''_2})$.
\end{description}

\noindent
The remaining cases, i.e.\ $s_0 \blackstretcheq \link{\noact}{\noact}s_0$ , $\link{\noact}{\noact}s_0  \blackstretcheq s_0$  and  
$s_0 \startchain{\alpha}\chainedlink{a}{a}\chainend{\beta} s'_0   \blackstretcheq
s_0 \startchain{\alpha}\chainedlink{a}{\noact}\chainedlink{\noact}{a}\chainend{\beta}s'_0$, 
have  similar proofs and are omitted.
\end{proof}

We recall Lemma~\ref{lemma:res_noact}
(the original appears on p.\,\pageref{lemma:res_noact}).
\lemmaquindici*

\begin{proof}
The proof derives from the definitions of the restriction $\restrict{}$ operator  and of the merge operator $\merges{}{}$, both defined  on link chains.
\begin{enumerate}
\item[(i)]
Obvious, as $\restrict{a}\alpha \neq \ \noact$ if $\alpha \neq \ \noact$.
\item[(ii)]
If $\restrict{a}s' = \bot$ it means that $a$ is unmatched in $s'$ and since $a$ does not appear in $s$ it remains unmatched in $\merges{s}{s'}$.
Otherwise, $a$ is matched in $s'$ and since solid links are preserved by $\merges{}{}$ it remains matched in $\merges{s}{s'}$, then renaming $a$ to $\silent$ before or after the merge does not change the result.
\item[(iii)]
Obvious as $\restrict{a}\restrict{b}\alpha = \restrict{b}\restrict{a}\alpha$ for any $\alpha$.
\end{enumerate}
\end{proof}

We recall Lemma~\ref{lemma:res}
(the original appears on p.\,\pageref{lemma:res}).
\lemmadiciassette*

\begin{proof}
We prove that the axioms, when applied in either direction, satisfy the property. Then, by transitivity, the thesis holds for all the elements in each equivalent class of $\blackstretcheq$.
The proof proceeds by cases on axioms of  $\blackstretcheq$.
\begin{description}
\item[{\bf case $[s_0  \blackstretcheq s_0 \link{\noact}{\noact}]$}]
Let $\restrict{a}s = s_0$ and  $s' = s_0\link{\noact}{\noact}$. 
Then, we set $s'' = s \link{\noact}{\noact}$, and it is immediate to verify that $a$ is matched in $s''$, as it is in $s$, thus we can write $s' = \restrict{a} s''$ (with $\restrict{a} s \blackstretcheq \restrict{a} s'')$ and we get that $s \blackstretcheq s''$.
\item[{\bf case $[s_0 \link{\noact}{\noact} \blackstretcheq s_0]$}]
Let $\restrict{a}s = s_0 \link{\noact}{\noact}$ and $s' = s_0$. Then it must exists $s''$ s.t.  $s =  s''\link{\noact}{\noact}$ with $s_0 = \restrict{a}s''$. The thesis follows as $s'=s_0 = \restrict{a}s''$ and clearly $s \blackstretcheq s''$.
\item[{\bf case $[s_0 \link{\noact}{\noact} s'_0 \blackstretcheq s_0 \startchain{\noact}\chainedlink{\noact}{\noact}\chainend{\noact}s'_0]$}] 
Let $\restrict{a}s = s_0 \link{\noact}{\noact} s'_0$ and 
$s' = s_0 \startchain{\noact}\chainedlink{\noact}{\noact}\chainend{\noact}s'_0$.
Then it must be 
$s = s_1 \link{\noact}{\noact} s_2$ for some $s_1$ and $s_2$ with 
$\restrict{a}s_1 = s_0$ and $\restrict{a}s_2 = s'_0$.
We set $s'' = s_1 \startchain{\noact}\chainedlink{\noact}{\noact}\chainend{\noact} s_2$, from which the thesis immediately follows.
\item[{\bf case $[s_0 \startchain{\alpha}\chainedlink{b}{\noact}\chainedlink{\noact}{b}\chainend{\beta} s'_0 \blackstretcheq s_0 \startchain{\alpha}\chainedlink{b}{b}\chainend{\beta}s'_0]$}] 
Let $\restrict{a}s = s_0 \startchain{\alpha}\chainedlink{b}{\noact}\chainedlink{\noact}{b}\chainend{\beta} s'_0$ and 
$s' = s_0 \startchain{\alpha}\chainedlink{b}{b}\chainend{\beta}s'_0$.
As $a$ cannot appear in $\restrict{a}s$, it must be the case that $a\neq b$
and $s = s_1 \link{\noact}{\noact} s_2$ for some $s_1$ and $s_2$ with 
$\restrict{a}s_1 =s_0 \link{\alpha}{b}$ and $\restrict{a}s_2 = \link{b}{\beta} s'_0$.
We set $s'' = s_1 s_2$, from which the thesis immediately follows.
\item[{\bf case $[s_0 \startchain{\alpha}\chainedlink{b}{b}\chainend{\beta}s'_0
\blackstretcheq   s_0 \startchain{\alpha}\chainedlink{b}{\noact}\chainedlink{\noact}{b}\chainend{\beta} s'_0]$}] Let $\restrict{a}s = s_0 \startchain{\alpha}\chainedlink{b}{b}\chainend{\beta}s'_0$ and $s' =s_0 \startchain{\alpha}\chainedlink{b}{\noact}\chainedlink{\noact}{b}\chainend{\beta} s'_0$.
As $a$ cannot appear in $\restrict{a}s$, it must be the case that $a\neq b$
and $s = s_1 s_2$ for some $s_1$ and $s_2$ with 
$\restrict{a}s_1 =s_0 \link{\alpha}{b}$ and $\restrict{a}s_2 = \link{b}{\beta} s'_0$.
We set $s'' = s_1 \link{\noact}{\noact} s_2$, from which the thesis immediately follows.
 
\end{description}
\noindent
We omit the remaining cases that are analogous.
\end{proof}

We recall Lemma~\ref{lemma:rename}
(the original appears on p.\,\pageref{lemma:rename}).
\lemmadiciannove*

\begin{proof} \
The proof of the points (i), (ii), (iii) derives from the definitions of the renaming function $\phi{}$, of the merge operator $\merges{}{}$, and of restriction operator $(\nu \ )$, all defined  on link chains.
\begin{itemize}
\item[(i)]
Obvious, as $\ell[\phi] \neq \ \link{\noact}{\noact}$ if $\ell \neq \ \link{\noact}{\noact}$.
\item[(ii)]
Let $s = \ell_1 ... \ell_n$ and 
$s' = \ell'_1 ... \ell'_n$, with $\ell_i = \link{\alpha_i}{\beta_i}$ and $\ell'_i = \link{\alpha'_i}{\beta'_i}$, for all $i \in [1,n]$.
Then, by definition of merge and renaming:
$$(s \bullet s') [\phi] = ((\merges{\ell_{1}}{\ell'_{1}})\cdots(\merges{\ell_{n}}{\ell'_{n}}))[\phi] = 
(\merges{\ell_{1}}{\ell'_{1}})[\phi]\cdots(\merges{\ell_{n}}{\ell'_{n}})[\phi] .$$
Since, for all $i \in [1,n]$,
$(\merges{\ell_{i}}{\ell'_{i}})[\phi] = 
\link{(\merges{\phi(\alpha_i)}{\phi(\alpha'_i)})}{(\merges{\phi(\beta_i)}{\phi(\beta'_i)})} =
\link{\phi(\alpha_i)}{\phi(\beta_i)} \bullet \link{\phi(\alpha'_i)}{\phi(\beta'_i)} =
(\merges{\ell_{i}[\phi]}{\ell'_{i}}[\phi])$, 
we can conclude that
$(s \bullet s') [\phi] = s[\phi] \bullet (s' [\phi])$.
\item[(iii)]
Let $s = \ell_1 ... \ell_n$  with $\ell_i = \link{\alpha_i}{\beta_i}$ then
$(\restrict{a}s)[\phi] = ((\restrict{a}\ell_1)\dots (\restrict{a}\ell_n))[\phi]$,
that amounts to
$((\restrict{a}\ell_1)[\phi]\dots (\restrict{a}\ell_n)[\phi])$.
Note that for any $\alpha$ we have 
$\phi(\restrict{a} \alpha) = \restrict{\phi(a)} \phi(\alpha)$.
In fact, if $\alpha = a$ then $\phi(\restrict{a} a) = \phi(\tau) = \tau$ and
$\restrict{\phi(a)} \phi(a) = \tau$.
If instead $\alpha \neq a$, then $\phi(\alpha) \neq \phi(a)$ (because $\phi$ is a bijection)
and thus
$\phi(\restrict{a} \alpha) = \phi(\alpha) = \restrict{\phi(a)} \phi(\alpha)$.
Since, for all $i \in [1,n]$, 
$(\restrict{a}\ell_i))[\phi] = (\restrict{a}\link{\alpha_i}{\beta_i})[\phi] =
\link{\restrict{a}\alpha_i}{\restrict{a}\beta_i}[\phi] = 
\link{\phi(\restrict{a}\alpha_i)}{\phi(\restrict{a}\beta_i)} = 
\link{\restrict{\phi(a)}\phi(\alpha_i)}{\restrict{\phi(a)}\phi(\beta_i)} =
\restrict{\phi(a)}(\link{\phi(\alpha_i)}{\phi(\beta_i)}) =
\restrict{\phi(a)}(\link{\alpha_i}{\beta_i}[\phi])$, 
we can conclude that 
$(\restrict{a} s)[\phi] = \restrict{\phi(a)} (s[\phi])$.
\item[(iv)]
The proof of (iv) derives from the compositionality of renaming functions, in particular,
$(\link{\alpha}{\beta})[\phi] [\psi] = \link{\phi(\alpha)}{\phi(\beta)}[\psi] = \link{\psi\circ \phi(\alpha)}{\psi\circ \phi(\beta)} = (\link{\alpha}{\beta})[\psi\circ \phi]$.
\item[(v)]
To prove (v), we prove that the axioms, when applied in either direction, satisfy the property. 
Then, by transitivity, the thesis holds for all the elements in each equivalent class of $\blackstretcheq$.
The proof proceeds by cases on axioms of  $\blackstretcheq$.
For the sake of simplicity, we show only one case.
\begin{description}
\item[{\bf case $[s_0 \blackstretcheq s_0 \link{\noact}{\noact}]$}]
Let $s=s_0$ and $s' = s_0 \link{\noact}{\noact}$. 
We have to show that $s[\phi] \blackstretcheq s'[\phi]$ i.e.~that
$s_0[\phi] \blackstretcheq s_0 \link{\noact}{\noact}[\phi]$, where
$s_0 \link{\noact}{\noact}[\phi] = s_0[\phi] \link{\noact}{\noact}$, since $\phi$ distributes over the single links.
Therefore, we obtain that 
$s_0[\phi] \blackstretcheq s_0[\phi] \link{\noact}{\noact}$.
\end{description}
\end{itemize}

\end{proof}

We recall Lemma~\ref{lemma:ren}
(the original appears on p.\,\pageref{lemma:ren}).
\lemmaventi*
\begin{proof}
Since $\phi$ is a bijection, we take its inverse $\phi^{-1}$ and let $s'' \defeq s'[\phi^{-1}]$. 
Then the thesis holds by Lemma~\ref{lemma:rename}(iv--v): 
$s' = s''[\phi]$ trivially holds, since $s''[\phi]= s'[\phi^{-1}][\phi] = s'$, and
$s \blackstretcheq s'[\phi^{-1}]$ holds because $s[\phi] \blackstretcheq s'$ implies
$s=s[\phi][\phi^{-1}] \blackstretcheq s'[\phi^{-1}]$.
\end{proof}

\subsection{Proofs of Section~\ref{abs-sem}}

We recall Lemma~\ref{lemma:stretcheqone}
(the original appears on p.\,\pageref{lemma:stretcheqone}).
\lemmatrentaquattro*

\begin{proof}
\item[(i)]
Take $s$ and let 
\begin{itemize}
\item $n$ be  the number of adjacent solid links in $s$, 
\item $m$ be the number of adjacent virtual links in $s$, 
\item $k$ be the number of virtual links at the extremes of $s$. 
\end{itemize}

For example, for $s=
\startchain{\noact}
\chainedlink{\noact}{\noact}
\chainedlink{\noact}{a}
\chainedlink{\silent}{\silent}
\chainedlink{b}{b}
\chainedlink{c}{\noact}
\chainedlink{\noact}{c}
\chainedlink{d}{\noact}
\chainedlink{\noact}{\noact}
\chainedlink{\noact}{e}
\chainend{\silent}$, we have 
$n=3$,
$m=4$
and 
$k=2$.

We prove the existence of $s'$ by induction on $v(s)=n+m+k$. 
\begin{itemize}
\item
For the base case, if $v(s)=0$ then $s$ is essential and we are done. 
\item
For the inductive case, suppose $v(s)>0$. Then at least one of $n,m,k$ is greater than $0$. 

If $n>0$,  then  there are two adjacent links in $s$ such as $\startchain{\alpha}\chainedlink{a}{a}\chainend{\beta}$ or $\startchain{\alpha}\chainedlink{\silent}{\silent}\chainend{\beta}$. In the former case, we can apply the last axiom of $\blackstretcheq$ (Definition~\ref{def:black}) 
to introduce a virtual link between the matched action $a$ and decrement by one the number of adjacent solid links.
In the latter case, we can apply the rightmost axiom of $\stretcheq$  (Definition~\ref{def:accordion}) to eliminate the matched $\silent$ actions
 and decrement by one the number of adjacent solid links.

If $m>0$, then there are two adjacent virtual links in $s$ and we can apply the top-right axiom of $\blackstretcheq$ (Definition~\ref{def:black}) to decrement by one the number of adjacent virtual links.

If $k>0$ we can apply one of the leftmost axioms of $\blackstretcheq$ (Definition~\ref{def:black}) to decrement by one the number of virtual links at the extremes. 

In all cases we get a chain $s''\stretcheq s$ with $v(s'')=v(s)-1$ and, by inductive hypothesis, there is an essential link chain $s'$ such that $s' \stretcheq s''$. Thus, by transitivity, we have $s' \stretcheq s$.
 \end{itemize}

\item[(ii)]
By contradiction, let $s = \ell_1 \link{\noact}{\noact}\ell_2 \dots \link{\noact}{\noact} \ell_n$ and $s' = \ell'_1 \link{\noact}{\noact}\ell'_2 \dots \link{\noact}{\noact} \ell'_m$ be two essential link  chains such that $s \stretcheq s'$ and $s \neq s'$. Without loss of generality, assume $s$ and $s'$ are chosen such that the length of $s$ is the minimal one for which such a counterexample exists. If $\ell_1 = \ell_1'$, then $\ell_2 \dots \link{\noact}{\noact} \ell_n$ and $\ell'_2 \dots \link{\noact}{\noact} \ell'_m$ would provide a shorter counterexample, contradicting the hypothesis of minimality for $n$. Thus it must be $\ell_1 \neq \ell_1'$. Let $\ell_1 = \link{\alpha_1}{\beta_1}$ and $\ell'_1 = \link{\alpha'_1}{\beta'_1}$. Now we can notice that the axioms for $\stretcheq$ preserves the leftmost non-virtual symbol of a link chain. Thus $\alpha_1 = \alpha'_1$, otherwise $s \stretcheq s'$ would not hold. Finally, we notice that any non-virtual symbol adjacent to a virtual link is preserved by the axioms. Thus $\beta_1 = \beta'_1$ and $\ell_1 = \ell'_1$ leading to a contradiction.

\end{proof}


We recall Lemma~\ref{lem:reswhite}
(the original appears on p.\,\pageref{lem:reswhite}).
\lemmatrentasei*

\begin{proof}
We prove that the property holds for the axioms of $\stretcheq$ when applied in each direction, then the fact that the property is preserved by the rules for equivalence is immediate. We prove only some cases; the remaining ones are similar.
\begin{description}
\item[case [$s_{1}\startchain{\alpha}\chainedlink{\silent}{\silent}\chainend{\beta}s_{2} \stretcheq s_{1}\link{\alpha}{\beta}s_{2}$]]
Let $s = s_{1}\startchain{\alpha}\chainedlink{\silent}{\silent}\chainend{\beta}s_{2}$ and 
$s' = s_{1}\link{\alpha}{\beta}s_{2}$. Since $\restrict{a}s\neq \bot$ it means that $a$ is matched in $s$ and thus it is matched in $s'$, which differs from $s$ only for the removal of $\silent$, and we put $s'' = s'$.
Then we have 
\begin{eqnarray*}
\restrict{a}s'' & = & (\restrict{a}s_1)\link{\restrict{a}\alpha}{\restrict{a}\beta}(\restrict{a}s_2)\\
& \stretcheq & (\restrict{a}s_1)\startchain{\restrict{a}\alpha}\chainedlink{\silent}{\silent}\chainend{\restrict{a}\beta}(\restrict{a}s_2) \\
& = & \restrict{a}s .
\end{eqnarray*}
\item[case [$s_1 \startchain{\alpha}\chainedlink{a}{a}\chainend{\beta} s_2 \stretcheq s_1  \startchain{\alpha}\chainedlink{a}{\noact}\chainedlink{\noact}{a}\chainend{\beta}s_2$]] Let $s= s_1 \startchain{\alpha}\chainedlink{a}{a}\chainend{\beta} s_2$ and $s' =  s_1 \startchain{\alpha}\chainedlink{a}{\noact}\chainedlink{\noact}{a}\chainend{\beta}s_2$.  We let $s'' = s \blackstretcheq s'$ and we are done, since $\restrict{a}s'' = \restrict{a}s$ is defined by hypothesis. 
\item[case [$s_1 \startchain{\alpha}\chainedlink{a}{\noact}\chainedlink{\noact}{a}\chainend{\beta}s_2 \stretcheq  s_1 \startchain{\alpha}\chainedlink{a}{a}\chainend{\beta} s_2$]] Let $s= s_1 \startchain{\alpha}\chainedlink{a}{\noact}\chainedlink{\noact}{a}\chainend{\beta}s_2$ and $s' = s_1 \startchain{\alpha}\chainedlink{a}{a}\chainend{\beta} s_2$. Since $\restrict{a}s$ is not defined, we are done.
\end{description}
\qed
\end{proof}

We recall Lemma~\ref{lem:mergewhite}
(the original appears on p.\,\pageref{lem:mergewhite}).
\lemmatrentasette*
\begin{proof}
We prove that the property holds for the axioms of $\stretcheq$, then the fact that the property is preserved by the rules for equivalence is immediate. We prove only two cases, the remaining ones are similar.
\begin{description}
\item[case [$s\startchain{\alpha}\chainedlink{\silent}{\silent}\chainend{\beta}s' \stretcheq s\link{\alpha}{\beta}s'$]]
We have $s_1 = s\startchain{\alpha}\chainedlink{\silent}{\silent}\chainend{\beta}s'$
and $s'_1 = s\link{\alpha}{\beta}s'$.
By definition of valid link, the links $\link{\alpha}{\silent}$ and $\link{\silent}{\beta}$ are solid, i.e. \
$\alpha \neq \ \noact$ and $\beta \neq \ \noact$.
Then, since $\merges{s_2}{s_{1}}\neq \bot$ we infer that  $s_2 = s''\startchain{\noact}\chainedlink{\noact}{\noact}\chainend{\noact}s'''$ for some $s'',s''$ such that 
$\merges{s''}{s}\neq \bot$, and $\merges{s'''}{s'}\neq \bot$.
Then we take $s'_2=s'' \link{\noact}{\noact}s''' \blackstretcheq s_2$ and $s''_1=s'_1$ and we are done, since  $\merges{s_2}{s_{1}} 
= 
(\merges{s''}{s})\startchain{\alpha}\chainedlink{\silent}{\silent}\chainend{\beta}(\merges{s'''}{s'}) 
\stretcheq 
(\merges{s''}{s})\link{\alpha}{\beta}(\merges{s'''}{s'})
=
\merges{s'_2}{s''_{1}}$.


\item[case [$s\startchain{\alpha}\chainedlink{a}{\noact}\chainedlink{\noact}{a}\chainend{\beta}s' \stretcheq s \startchain{\alpha}\chainedlink{a}{a}\chainend{\beta}s'$]]
We have $s_1 = s\startchain{\alpha}\chainedlink{a}{\noact}\chainedlink{\noact}{a}\chainend{\beta}s'$
and $s'_1 = s\startchain{\alpha}\chainedlink{a}{a}\chainend{\beta}s'$.
Then, there are two possibilities: either
(a) $s_2 = s''\startchain{\noact}\chainedlink{\noact}{\noact}\chainedlink{\noact}{\noact}\chainend{\noact}s'''$, or (b)
$s_2 = s''\startchain{\noact}\chainedlink{\noact}{a}\chainedlink{a}{\noact}\chainend{\noact}s'''$
with
$\merges{s''}{s}\neq \bot$, and $\merges{s'''}{s'}\neq \bot$.
In the case (a), we take  $s'_2 = s''\startchain{\noact}\chainedlink{\noact}{\noact}\chainend{\noact}s''' \blackstretcheq s_2$  and $s''_1 = s'_1$ and we are done, since
$ \merges{s_2}{s_1} =(\merges{s''}{s})\startchain{\alpha}\chainedlink{a}{\noact}\chainedlink{\noact}{a}\chainend{\beta}(\merges{s'''}{s'}) \stretcheq (\merges{s''}{s})\startchain{\alpha}\chainedlink{a}{a}\chainend{\beta}(\merges{s'}{s'''})  = \merges{s'_2}{ s''_1}$.
In the case (b) we take $s'_2 = s_2$ and $s''_1 = s\startchain{\alpha}\chainedlink{a}{\noact}\chainedlink{\noact}{a}\chainend{\beta}s' \blackstretcheq s'_1$ and we are done, since 
$ \merges{s_2}{s_1} =(\merges{s''}{s})\startchain{\alpha}\chainedlink{a}{a}\chainedlink{a}{a}\chainend{\beta}(\merges{s'''}{s'}) \stretcheq (\merges{s''}{s})\startchain{\alpha}\chainedlink{a}{a}\chainedlink{a}{a}\chainend{\beta}(\merges{s'}{s'''})  = \merges{s'_2}{ s''_1}$.
\end{description}

\end{proof}

We recall Lemma~\ref{lem:reduce}
(the original appears on p.\,\pageref{lem:reduce}).
\lemmatrentotto*

\begin{proof}
The proof is similar to the ones of Lemma~\ref{lemma:rename} (point $v$) since, 
by definition $\phi(\silent) = \silent$ and $\phi({\noacts})\ =\ \noacts$, then the equivalence relation 
$\stretcheq$ is not affected by $\phi$.
\end{proof}

We recall Theorem~\ref{theo:CNAcong}
(the original appears on p.\,\pageref{theo:CNAcong}).
\theoremquarantatre*
\begin{proof}
We complete here the proof outlined at p.\,\pageref{theo:CNAcong}, by
giving the details of the case for recursion.
\begin{description}
\item[Recursion]
Let $E$ and $F$ be two processes that invoke the process identifier $X$.
Assume that for any process $P$ we have  $E\{P \slash X\} \simnet F\{P \slash X\}$.
We want to prove that, given  the process definitions
$A\defeq E\{A \slash X\}$, $B\defeq F\{B \slash X\}$, then $A \simnet B$.  
The proof proceeds by showing that:
\begin{enumerate}
\item If $A \defeq Q$ is a process definition, then $A  \simnet Q$.
\item Given the process definitions $A\defeq E\{A \slash X\}$ and $B\defeq F\{B \slash X\}$, for any process $G$ that invokes $X$ we have $G\{A \slash X\} \simnet 
G\{B \slash X\}$.
\end{enumerate}
Then, we have $A \simnet E\{A \slash X \} \simnet E\{B \slash X\}  \simnet F\{B \slash X\} \simnet B$.
The proof of~(1) is immediate by rule \textit{(Ide)}, as $A$ and $Q$ have exactly the same transitions, while the proof of~(2) proceeds in the standard way exploiting induction on derivations as detailed below.


Let $\mathbf{R}_{ctx} \defeq \{(G\{A \slash X \} ,G\{B \slash X\} ) \mid \mbox{$G$ is a process that possibly invokes $X$}\}$. 
Let $\mathbf{R}_{upto} \defeq\; \simnet\, \circ\; \mathbf{R}_{ctx}\; \circ\, \simnet$. 
Note that  $\mathbf{R}_{ctx}$ includes the identity relation when taking $G$ with no occurrence of $X$.
Moreover,  $\mathbf{R}_{ctx} \subseteq \mathbf{R}_{upto}$, 
$\simnet \subseteq \mathbf{R}_{upto}$ and $\mathbf{R}_{upto}\; \circ\, \simnet = \mathbf{R}_{upto}$ because 
$\simnet$ is an equivalence relation (and thus transitively closed).
We prove that $\mathbf{R}_{upto}$ is a network bisimulation.
To this aim, it is enough to consider a generic pair of processes $(G\{A \slash X \} ,G\{B \slash X\})$ in $\mathbf{R}_{ctx}$ and prove that whenever $G\{A \slash X \}\xrightarrow{s} P'$ 
then there are some $s'$ and $Q'$ such that $G\{B \slash X\}\xrightarrow{s'} Q'$, $\reduce{s}=\reduce{s'}$ and $(P',Q')\in \mathbf{R}_{upto}$.
We proceed by induction on the derivation of the transition $G\{A \slash X \}\xrightarrow{s} P'$, 
by considering the possible shapes of $G$.
\begin{description}
\item[$G=X$:] 
We have $G\{A \slash X \} = A$. 
Since $A\xrightarrow{s} P'$ and $A\defeq E\{A \slash X \}$, 
it means that $E\{A \slash X \} \xrightarrow{s} P'$ with a shorter derivation than $A\xrightarrow{s} P'$.
Hence, by inductive hypothesis, there are $s''$ and $Q''$ such that 
$E\{B \slash X\}\xrightarrow{s''} Q''$, $\reduce{s}=\reduce{s''}$ and $(P',Q'')\in \mathbf{R}_{upto}$.
Since  $E\{P \slash X\} \simnet F\{P \slash X\}$ for any process $P$, we have in particular $E\{B \slash X\} \simnet F\{B \slash X\}$.
So there are $s'$ and $Q'$ such that 
$F\{B \slash X\}\xrightarrow{s'} Q'$, $\reduce{s''}=\reduce{s'}$ and $Q'' \simnet Q'$.
Since $B\defeq F\{B \slash X\}$, by applying rule \textit{(Ide)} we have $B\xrightarrow{s'} Q'$.
We conclude by noting that $G\{B \slash X\} = B$, $\reduce{s}=\reduce{s''}=\reduce{s'}$ and that $(P',Q') \in \mathbf{R}_{upto}\, \circ \simnet = \mathbf{R}_{upto}$.

\item[$G=\ell.G'$:]
We have $G\{A \slash X \}= \ell.(G'\{A \slash X \})$ and thus $P' =  G'\{A \slash X \}$.
Moreover $G\{B \slash X\} = \ell.(G'\{B \slash X\}) \xrightarrow{s} G'\{B \slash X\}$ and by definition of $\mathbf{R}_{upto}$
we have $(G'\{A \slash X \},G'\{B \slash X\})\in \mathbf{R}_{ctx} \subseteq \mathbf{R}_{upto}$.

\item[$G=G_1 + G_2$:]
We have $G\{A \slash X \} = G_1\{A \slash X \} + G_2\{A \slash X \}$.
Since we have $G\{A \slash X \}\xrightarrow{s} P'$ there are two possibilities: 
either $G_1\{A \slash X \}\xrightarrow{s} P'$ or $G_2\{A \slash X\}\xrightarrow{s} P'$ (with shorter derivations).
Without loss of generality, let us consider just the first case.
By inductive hypothesis, there are $s'$ and $Q'$ such that 
$G_1\{B \slash X\}\xrightarrow{s'} Q'$, $\reduce{s}=\reduce{s'}$ and $(P',Q') \in \mathbf{R}_{upto}$.
Then, by rule \textit{(Lsum)},  $G\{B \slash X\}= G_1\{B \slash X\} + G_2\{B \slash X\}\xrightarrow{s'} Q'$.

\item[$G=\restrict{a}G'$:]
We have $G\{A \slash X \} = \restrict{a}(G'\{A \slash X\})$. 
Thus $P'=\restrict{a}P''$ and $s=\restrict{a}s''$ for some $P''$ and $s''$ such that 
$G'\{A \slash X \}\xrightarrow{s''} P''$ (with a shorter derivation).
By inductive hypothesis, there are $s''_1,Q''$ such that 
$G'\{B \slash X\}\xrightarrow{s''_1} Q''$, $\reduce{s''}=\reduce{s''_1}$ and $(P'',Q'') \in \mathbf{R}_{upto}$.
By Lemma~\ref{lem:reswhite}, there exists $s''_2 \blackstretcheq s''_1$ such that
 $\restrict{a}s''_2\neq \bot$ and $\restrict{a}s''_2 \stretcheq \restrict{a}s''$.
By the Accordion Lemma~\ref{lem:stretchlts}, $G'\{B \slash X\}\xrightarrow{s''_2} Q''$.
Then, we take $s' = \restrict{a}s''_2$ and $Q' = \restrict{a}Q''$ and by rule \textit{(Res)}
$G\{B \slash X\}= \restrict{a}(G'\{B \slash X\})\xrightarrow{s'} Q'$.
Clearly $\reduce{s}=\reduce{s'}$.
To see that $(P',Q')\in \mathbf{R}_{upto}$ we note that by $(P'',Q'') \in \mathbf{R}_{upto}$ there is some $H$ such that 
$P'' \simnet H\{A \slash X \}$ and $H\{B \slash X\} \simnet Q''$. Then, as $\simnet$ is a congruence w.r.t.\ restriction, 
$P' = \restrict{a}P'' \simnet \restrict{a}H\{A \slash X \}$ and $\restrict{a}H\{B \slash X\} \simnet \restrict{a}Q'' = Q'$ and we are done.

\item[\protect{$G=G'[\phi]$}:]
This case in analogous to the previous one and thus omitted.

\item[$G=G_1 | G_2$:]
We have $G\{A \slash X \} = G_1\{A \slash X \} | G_2\{A \slash X\}$.
Since $G\{A \slash X \}\xrightarrow{s} P'$ we have three cases: 
(i) $G_1\{A \slash X \}\xrightarrow{s} P'_1$  and $P' = P'_1 | G_2\{A \slash X \}$, or
(ii) $G_2\{A \slash X \}\xrightarrow{s} P'_2$  and $P' = G_1\{A \slash X \} | P'_2$, or
(iii) $G_1\{A \slash X \}\xrightarrow{s_1} P'_1$, $G_2[\{A \slash X \}\xrightarrow{s_2} P'_2$  and $P' = P'_1 | P'_2$ with $s = \merges{s_1}{s_2}$.
Without loss of generality, let us focus on the third case, which is the more involved.
By inductive hypothesis, there are $s''_i,Q'_i$ with
$G_i\{B \slash X\}\xrightarrow{s''_i} Q'_i$, $\reduce{s_i}=\reduce{s''_i}$ and $(P'_i,Q'_i) \in \mathbf{R}_{upto}$ for $i=1,2$.
By Lemma~\ref{lem:mergewhite}, we know that $s''_{1}$ and $s''_{2}$ can be stretched respectively to $s'_{1}\blackstretcheq s''_{1}$ and $s'_{2}\blackstretcheq s''_2$ so that $\merges{s'_{1}}{s'_{2}}$ is defined and
$\reduce{\merges{s_{1}}{s_{2}}} = \reduce{\merges{s'_{1}}{s'_{2}}}$.
By the Accordion Lemma~\ref{lem:stretchlts},
$G_i\{B \slash X\}\xrightarrow{s'_i} Q'_i$ for $i=1,2$ and 
by rule \textit{(Par)}, we get
$G\{B \slash X\} = G_1\{B \slash X\} | G_2\{B \slash X\}\xrightarrow{s'} Q'$ with $s' =\merges{s'_1}{s'_2}$ and $Q'= Q'_1|Q'_2$.
To see that $(P',Q')\in \mathbf{R}_{upto}$ we note that, for $i=1,2$, by $(P'_i,Q'_i) \in \mathbf{R}_{upto}$ there is some $H_i$ such that 
$P'_i \simnet H_i\{A \slash X \}$ and $H_i\{B \slash X\} \simnet Q'_i$. Then, as $\simnet$ is a congruence w.r.t. parallel composition, 
$P' = P'_1|P'_2 \simnet H_1\{A \slash X \}|H_2\{A \slash X \}$ and $H_1\{B \slash X\}|H_2\{B \slash X\} \simnet Q'_1|Q'_2 = Q'$ and we are done.

\item[$G=C$:]
The simplest case is when $G$ is a constant $C$ associated with a definition $C\defeq R$.
In fact, we have $G\{A \slash X \} = C = G\{B \slash X\}$ 
and we conclude by taking $s'=s$ and $Q'=P'$.
\end{description}
\end{description}
\end{proof}

We recall Lemma~\ref{lem:linksubst}
(the original appears on p.\,\pageref{lem:linksubst}).
\lemmaquarantacinque*
\begin{proof}
The proof proceeds by cases on the axioms of $\blackstretcheq$ and $\stretcheq$, see Definitions~\ref{def:black} and~\ref{def:accordion}. We prove only one case, the remaining ones are similar.
\begin{description}
\item[case [$s\startchain{\alpha}\chainedlink{\silent}{\silent}\chainend{\beta}s' \stretcheq s\link{\alpha}{\beta}s'$]] We have  $s_1 = s\startchain{\alpha}\chainedlink{\silent}{\silent}\chainend{\beta}s'$ and $s_2 = s\link{\alpha}{\beta}s'$. Let $a$, $b$ two channel names then, by definition of substitution, we have 
\begin{eqnarray*}
s_1\{\subst{b}{a}\} & = & (s\startchain{\alpha}\chainedlink{\silent}{\silent}\chainend{\beta}s')\{\subst{b}{a}\}\\
& = & (s[\{\subst{b}{a}\})\ (\link{\alpha}{\silent}\{\subst{b}{a}\})\ (\link{\silent}{\beta}\{\subst{b}{a}\})\ (s'\{\subst{b}{a}\})\\
& = & (s\{\subst{b}{a}\})\ \startchain{\alpha \{\subst{b}{a}\}}\chainedlink{\silent}{\silent}\chainend{\beta\{\subst{b}{a}\}}\ (s'\{\subst{b}{a}\})\\ 
& \stretcheq& (s\{\subst{b}{a}\})\ \startchain{\alpha \{\subst{b}{a}\}}\chainend{\beta\{\subst{b}{a}\}}\ (s'\{\subst{b}{a}\}) \\
& = & ( s\link{\alpha}{\beta}s') \{\subst{b}{a}\} \\
& = & s_2 \{\subst{b}{a}\}.
\end{eqnarray*} 
\end{description}
\end{proof}

We recall Lemma~\ref{lem:esempio}
(the original appears on p.\,\pageref{lem:esempio}).
\lemmacinquantuno*

\begin{proof} We prove the two implications separately.
\begin{enumerate}
\item
The proof is by structural induction on the composite infrastructure $R(\widetilde{a},\widetilde{b})$.

If it is a basic infrastructure 
$R(\widetilde{a},\widetilde{b}) = \ell_1.R(\widetilde{a},\widetilde{b}) + ... + \ell_k.R(\widetilde{a},\widetilde{b})$, then it must be the case that $R' = R(\widetilde{a},\widetilde{b})$ and $s \blackstretcheq \ell_h = \link{a_{i_h}}{b_{j_h}}$ for some $h\in [1,k]$. Clearly $||s|| = 1$ and in fact there is a path of length $1$ from $a_{i_h}$ to $b_{j_h}$ in the graph $\mathcal{G}(R(\widetilde{a},\widetilde{b}))$.

If it is the composition
$$
\restrict{\widetilde{c}}(Q(\widetilde{a},\widetilde{c})\ |\ S(\widetilde{c},\widetilde{b}))
$$
of two infrastructures, then it must be the case that $s = \restrict{\widetilde{c}}(\merges{s_1}{s_2})$ for some $s_1$ and $s_2$ such that there exists $Q'$ and $S'$ with 
$Q(\widetilde{a},\widetilde{c}) \xrightarrow{s_1} Q'$,  
$S(\widetilde{c},\widetilde{b}) \xrightarrow{s_2} S'$ and
$R' = \restrict{c}(Q'\ |\ S')$.
Then, by inductive hypotheses, we know that
\begin{itemize}
\item $Q' = Q(\widetilde{a},\widetilde{c})$ and there exist two nodes $a_i\in \widetilde{a}$ and $c_h\in \widetilde{c}$ of the graph $\mathcal{G}(Q(\widetilde{a},\widetilde{c}))$ (and thus also in $\mathcal{G}(R(\widetilde{a},\widetilde{b}))$) such that $s_1 \stretcheq \link{a_i}{c_{h_1}}$ and there is a path from $a_i$ to $c_{h_1}$ whose size is $||s_1||$.
\item $S' = S(\widetilde{c},\widetilde{b})$ and there exist two nodes $c_{h_2}\in \widetilde{c}$ and $b_j\in \widetilde{b}$ of the graph $\mathcal{G}(S(\widetilde{c},\widetilde{b}))$ (and thus also in $\mathcal{G}(R(\widetilde{a},\widetilde{b}))$) such that $s_2 \stretcheq \link{c_{h_2}}{b_j}$ and there is a path from $c_{h_2}$ to $b_j$ whose length is $||s_2||$.
\end{itemize}
Since channels $\widetilde{c}$ are restricted and $\restrict{c}(\merges{s_1}{s_2})$ is well defined, it must be the case that $h_1 = h_2$ and $\merges{s_1}{s_2} \stretcheq \startchain{a_i}\chainedlink{c_{h_1}}{c_{h_1}}\chainend{b_j}$. Therefore 
$R' = \restrict{\widetilde{c}}(Q'\ |\ S') = \restrict{\widetilde{c}}(Q(\widetilde{a},\widetilde{c})|S(\widetilde{c},\widetilde{b})) = R(\widetilde{a},\widetilde{b})$,
$s = \restrict{\widetilde{c}}(\merges{s_1}{s_2}) \stretcheq \link{a_i}{b_j}$, $||s|| = ||s_1|| + ||s_2||$ and the two paths from $a_i$ to $c_{h_1}$ and from $c_{h_1}$ to $b_j$ can be composed to form a path from $a_i$ to $b_j$ whose length is exactly $||s||$. 

\item
The proof is by structural induction on the composite infrastructure $R(\widetilde{a},\widetilde{b})$.

If it is a basic infrastructure 
$R(\widetilde{a},\widetilde{b}) = \ell_1.R(\widetilde{a},\widetilde{b}) + ... + \ell_k.R(\widetilde{a},\widetilde{b})$ then the path from $a_i$ to $b_j$ in the graph must have length one and be in correspondence to one of the links offered by $R(\widetilde{a},\widetilde{b})$.

If it is the composition 
$$
\restrict{\widetilde{c}}(Q(\widetilde{a},\widetilde{c})\ |\ S(\widetilde{c},\widetilde{b}))
$$
of two infrastructures, then it must be the case that the path from $a_i$ to $b_j$ with length $n$ can be split in two parts: from $a_i$ to some $c_h$ (contained in the graph $\mathcal{G}(Q(\widetilde{a},\widetilde{c}))$) and from $c_h$ to $b_j$ (contained in the graph $\mathcal{G}(S(\widetilde{c},\widetilde{b}))$), respectively with lengths $n_1$ and $n_2$ such that $n=n_1+n_2$. Then, by the inductive hypotheses, it must be the case that
\begin{itemize}
\item $Q(\widetilde{a},\widetilde{c}) \xrightarrow{s_1} Q(\widetilde{a},\widetilde{c})$  with $s_1 \stretcheq \link{a_i}{c_h}$ and $||s_1|| = n_1$.
\item$S(\widetilde{c},\widetilde{b}) \xrightarrow{s_2} S(\widetilde{c},\widetilde{b})$  with $s_2 \stretcheq \link{c_h}{b_j}$ and $||s_2|| = n_2$.
\end{itemize}
Then we can find two suitable chains $s'_1 \blackstretcheq s_1$ and $s'_2 \blackstretcheq s_2$ such that $\merges{s'_1}{s'_2}$ is well defined and $\merges{s'_1}{s'_2} \stretcheq \startchain{a_i}\chainedlink{c_h}{c_h}\chainend{b_j}$. Therefore we take $s = \restrict{\widetilde{c}}(\merges{s'_1}{s'_2}) \stretcheq \link{a_i}{b_j}$ with $||s|| = ||s'_1|| + ||s'_2|| = ||s_1|| + ||s_2|| = n_1 + n_2 = n$ and by the rules of the operational semantics we have $R(\widetilde{a},\widetilde{b}) \xrightarrow{s} R(\widetilde{a},\widetilde{b})$.
\end{enumerate}
\end{proof}

\end{document}